\documentclass[12pt]{article}
\usepackage{amsmath}
\usepackage{graphicx} 
\usepackage{microtype}
\usepackage{inputenx}
\usepackage[english]{babel}
\usepackage{curve2e}
\usepackage{csquotes}
\usepackage[backend=biber, style=authoryear, citestyle=authoryear]{biblatex}
\usepackage{mathptmx}
\usepackage{helvet}
\usepackage{xcolor}
\usepackage{courier}
\usepackage{type1cm}
\usepackage{makeidx}
\usepackage{graphics}
\usepackage{multicol} 
\usepackage{amsmath}
\usepackage{booktabs}
\usepackage{float}
\usepackage[bottom]{footmisc}
\usepackage{amssymb}
\usepackage{lineno}
\usepackage{multirow}
\usepackage{wasysym}
\usepackage{tikz}
\usepackage{caption}
\usepackage{subfigure}
\usepackage{amsfonts}
\usepackage{url}
\usepackage{fancyhdr}
\usepackage{mathtools}
\usepackage{units}
\usepackage[T1]{fontenc} 
\usepackage{verbatim}
\usepackage{listings}
\usepackage{xcolor}
\usepackage{footnote}
\usepackage{enumerate}
\usepackage{ragged2e}
\usepackage{array}
\usepackage{bm}
\usepackage{calrsfs}
\usepackage{rotating}
\usepackage{tabulary}
\usepackage{blkarray}
\usepackage{bigstrut}
\usepackage{gauss}
\usepackage[ruled,vlined]{algorithm2e}
\usepackage[title]{appendix}
\usepackage{dirtytalk}
\usepackage[colorlinks=true, urlcolor=blue, linkcolor=red]{hyperref}

\newtheorem{theorem}{Theorem}[section]
\newtheorem{proof}{Proof}[section]
\newtheorem{definition}{Definition}[section]
\newcommand{\pkg}[1]{{\normalfont\fontseries{b}\selectfont #1}}
\let\proglang=\textsf
\let\code=\texttt

\newcommand{\blind}{0}

\begin{document}

\def\spacingset#1{\renewcommand{\baselinestretch}%
{#1}\small\normalsize} \spacingset{1}

%%%%%%%%%%%%%%%%%%%%%%%%%%%%%%%%%%%%%%%%%%%%%%%%%%%%%%%%%%%%%%%%%%%%%%%%%%%%%%

\if0\blind
{
  \title{\bf Unbiased mixed variables distance}
  \author{Michel van de Velden\\
    Econometric Institute, Erasmus University Rotterdam\\
    and \\
    Alfonso Iodice D'Enza \\
    Department of Political Sciences, University of Naples Federico II\\
    and \\
    Angelos Markos \\
    Department of Primary Education, Democritus University of Thrace\\
    and \\
    Carlo Cavicchia \\
    Econometric Institute, Erasmus University Rotterdam}
  \maketitle
} \fi

\if1\blind
{
  \bigskip
  \bigskip
  \bigskip
  \begin{center}
    {\LARGE\bf Title}
\end{center}
  \medskip
} \fi

\bigskip
\begin{abstract}
Defining a distance in a mixed setting requires the quantification of observed differences of variables of different types and of variables that are measured on different scales. There exist several proposals for mixed variable distances, however, such distances tend to be \say{biased} towards specific variable types and measurement units. That is, the variable types and scales influence the contribution of individual variables to the overall distance. In this paper, we define \say{unbiased} mixed variable distances for which the contributions of individual variables to the overall distance are not influenced by measurement types or scales. We define the relevant concepts to quantify such biases and we provide a general formulation that can be used to construct unbiased mixed variable distances.
\end{abstract}

\noindent%
{\it Keywords:}  Data visualization, distance-based methods, cluster analysis
\vfill

\newpage
\spacingset{1.75} 
\section{Introduction}
In many fields of study, it is necessary to measure the distance or similarity between observations or data points. These measurements are crucial for various tasks, including data visualization, clustering and classification. A significant challenge arises when a dataset consists of mixed variables, that is, when the observed variables differ in type or scale \parencite{BMMFRBC:2023}. Mixed variable datasets can consist of both numerical variables (continuous or discrete) and categorical variables (binary, ordinal, or nominal). These datasets are often encountered in machine learning and scientific computing tasks \parencite{Luo:2024, Hermes:2024}.

In a purely numerical context distance can be based on the sizes of the observed differences. For categorical variables, the situation is more complex as one only observes whether categories are equivalent or not. However, it may be appropriate to quantify differences between categories of the same variable differently. For example, one could argue that the dissimilarity between the colors purple and blue is smaller than the dissimilarity between yellow and blue. Similarly, the number of categories could play a role in quantifying differences. If there are many categories, there are more ways to differ. Many proposals exist to quantify distance for categorical variables; see, e.g. \textcite{BCK:2008,SR:2019,vandeveldenetal2024}. The complexity arising from quantifying differences and the impact that these choices have in a clustering context when working with categorical data is thoroughly demonstrated by \textcite{Amiri:2018}.

For mixed variables, the situation is even more complex, as this requires a method of calculating distance that can handle the diversity in variable types effectively, i.e., the aggregation over the variables should not be influenced by types or scales. One of the most recognized and widely used methods for calculating mixed variable distances is based on Gower’s general coefficient of similarity \parencite{Gower:1971}. This measure, which we refer to as Gower's distance, provides an elegant and simple way to measure the distance between observations based on measurement on numerical and categorical variables. Although alternative proposals that allow distance calculations in mixed variable contexts exist \parencite[see, e.g.,][]{AD:2007, Huang1998, van2019distance}, Gower's proposal -- implemented in the \pkg{gower} \proglang{R} package \parencite{gower_pkg} and in the \code{daisy} function available through the \pkg{cluster} \proglang{R} package \parencite{cluster_pkg} -- remains a popular choice.

It has been observed that biases can arise when computing distances in a multivariate context. For example, in numerical data, measurement scales can influence distances. Similarly, for categorical variables, the number of categories can affect the calculation of distances. Additionally, in settings with mixed variable types, it is important to control the balance between these types. For example, the Gower's distance is known to be biased toward the contribution of categorical variables \parencite{HennigLiao2013,foss2016semiparametric}. One way to overcome potential biases is to weigh the distances resulting from the numerical and categorical variables differently \parencite{Huang1998,de2021towards}. Alternatively, different ways to scale the numerical and categorical variables may be introduced \parencite{HennigLiao2013}. 

In this paper, we consider the problem of biased mixed variables distance. In particular, we propose a general formulation for multivariate mixed variable distances that can be used to construct unbiased mixed variable distances. To achieve this, we consider two properties: 1) \textit{multivariate additivity} and 2) \textit{commensurability}. Multivariate additivity means that the multivariate distance is the sum of univariate distances. The commensurability requirement ensures that the univariate contributions to the overall distance are commensurable (i.e., corresponding in size, extent, amount or degree). We show that existing mixed variable distances lead to biased distances, and we derive the magnitude of such biases. We evaluated several specific mixed variable distances and show that our new unbiased mixed variable distances provide a more objective starting point in an exploratory distance-based data analysis. 

This paper is organized as follows. After introducing the two desirable properties for multivariate mixed variable distances in Section \ref{properties}, we propose a general multivariate mixed variable distance formulation that can be used to construct mixed variable distances that satisfy these properties. In Sections \ref{sect:numdist} and \ref{sect:catdist}, we show how several common distance and scaling options -- for numerical and categorical variables -- violate the properties of Section \ref{properties} and hence may lead to biased mixed variable distances. In Section \ref{sect:unbiased} we show, using our general distance formulation, how unbiased mixed variable distances can be constructed. Using a simulation study and an empirical application, we illustrate our findings in Sections \ref{sect:simulation} and \ref{sect:illustration}, and end the article with our conclusions. 

\section{Properties of a multivariate mixed variable distance} \label{properties}
A multivariate mixed variable distance requires the quantification and aggregation of distances for different variables. Aggregation concerns \say{joining} or the accumulation of variable-specific distances. For this aggregation, we only consider distances that are \say{additive}. We define multivariate additivity as the property that in a $p-$dimensional setting, the distance can be expressed as the sum of $p$ distances. That is,
\begin{definition}[Multivariate Additivity] \label{DefMA}
Let $\mathbf{x}_i=\left(x_{i1}, \dots, x_{ip}\right)$ denote a $p-$dimensional vector. A distance function $d\left(\mathbf{x}_i,\mathbf{x}_l\right)$ between observations $i$ and $l$ is multivariate additive if
  \begin{equation*}
      d\left(\mathbf{x}_i,\mathbf{x}_l\right)=\sum_{j=1}^{p} d_j\left(\mathbf{x}_i,\mathbf{x}_l\right),
  \end{equation*}  
  where $d_j\left(\mathbf{x}_i,\mathbf{x}_l\right)$ denotes the $j-$th variable specific distance.
\end{definition}
Note that the variable-specific distances $d_j$ are functions of the $p-$dimensional vectors. These functions can be used to take into account associations between the $p$ variables. 

The Manhattan distance for $d\left(\mathbf{x}_i,\mathbf{x}_l\right)$ satisfies the additivity property, whereas the Euclidean distance, which takes the square root of the sum of squared differences, does not. Euclidean distance implicitly assumes that the variables jointly span a space. It is then defined by considering the length of a straight line connecting the points in this high-dimensional space. In a multivariate context, the contribution per variable for the Euclidean distance decreases as more variables are included. 

By defining the overall multivariate distance $d\left(\mathbf{x}_i,\mathbf{x}_l\right)$ as the sum of variable-specific distances, each variable-specific distance $d_j\left(\mathbf{x}_i,\mathbf{x}_l\right)$ is treated equivalently. This requires these distances to be on similar or equivalent scales. To achieve this we introduce an additional property: commensurability. In particular, we define distances as commensurable if the mean variable-specific distance between observations is constant across variables. That is,

\begin{definition}[Commensurability]  \label{DefCom}
Let $\bm{X}_i =\left(X_{i1}, \dots, X_{ip}\right)$ denote a $p-$dimensional random variable corresponding to observation $i$. Furthermore, let $d_{j}$ denote the distance function corresponding to the $j-$th variable. We have commensurability if,
for all $j$, and $i \neq l$, 
\begin{equation*}
    E[d_{j}(\bm{X}_i, \bm{X}_l)] = c,
\end{equation*}
where $c$ is some constant of which the actual value is of no importance. 
\end{definition}

In Definitions \ref{DefMA} and \ref{DefCom},  we allow for variable-type-specific distance measures. In a mixed variable type setting, such flexibility is essential as the numerical and categorical variables require different distance functions. 

\section{A general multivariate mixed variable distance} \label{proposal}

To construct a mixed distance that satisfies multivariate additivity and commensurability, as defined in Definitions \ref{DefMA} and \ref{DefCom}, we introduce a general multivariate mixed variable distance formulation.

Let $\mathbf{x}_i=\left(x_{i1}, \dots, x_{ip}\right)$ denote a $p-$dimensional vector. We define the distance between observations $i$ and $l$ as
\begin{equation}\label{genmixeddist}
    d\left(\mathbf{x}_i,\mathbf{x}_l\right)=\sum_{j=1}^{p} d_j\left(\mathbf{x}_i,\mathbf{x}_l\right)=\sum_{j=1}^{p} w_j\delta_{j}\left(\mathbf{x}_i,\mathbf{x}_l\right),
\end{equation}
where $\delta_j$ is a function quantifying the dissimilarity between observations $i$ and $l$ based on the $p-$dimensional variables, and $w_j$ is a weight corresponding to each of these functions. 

To distinguish numerical from categorical variables, we decompose the overall distance with respect to the variable types. For a numerical variable $j_n$, the distance is given by 
\begin{equation}\label{NumDist}
    d_{j_n}\left(\mathbf{x}_i,\mathbf{x}_l\right) = w_{j_n} \delta^n_{j_n}\left(\mathbf{x}^n_i,\mathbf{x}^n_l\right),
\end{equation}
where $\mathbf{x}^n_{i} $ denotes the vector of $p_n$ numerical elements of $\mathbf{x}_i$, $\delta^n_{j_n}$ is a function quantifying the dissimilarity between observations on the $j_n-$th numerical variable, and $w_{j_n}$ is a scale value, or weight, for the $j_n-$th variable. Similarly, for a categorical variable  $j_c$, the distance is
\begin{equation}\label{CatDist}
    d_{j_c}\left(\mathbf{x}_i,\mathbf{x}_l\right) = w_{j_c} \delta^c_{j_c}\left(\mathbf{x}^c_i,\mathbf{x}^c_l\right),
\end{equation}
where $\mathbf{x}^c_{i} $ 
denotes the vector of $p_c$ categorical elements of $\mathbf{x}_i$, $\delta^c_{j_c}$ 
corresponds to the category dissimilarity between the categories chosen by subjects $i$ and $l$ for categorical variable $j_c$, and $w_{j_c}$ is a scale value, or weight, for the $j_c-$th variable. Hence, the category dissimilarities quantify the differences between categories. This decomposition avoids using the full mixed vector of observations as input for distance calculations. As a result, the distances for the numerical and categorical variables depend only on the values of the variables of the same type. Although a more general formulation, where interaction between the variables of different types is theoretically appealing, defining and quantifying such interactions is challenging and we do not pursue this here. 

Combining Equations (\ref{genmixeddist}),(\ref{NumDist}) and (\ref{CatDist}), leads to the following formulation for mixed distance between observations $i$ and $l$:
\begin{equation}\label{genmixeddist_formula}
    d\left(\mathbf{x}_i,\mathbf{x}_l\right)=\sum_{j=1}^{p} w_j\delta_{j}\left(\mathbf{x}_i,\mathbf{x}_l\right)=\sum_{j_n=1}^{p_n} w_{j_n} \delta^n_{j_n}\left(\mathbf{x}^n_i,\mathbf{x}^n_l\right)+ \sum_{j_c=1}^{p_c} w_{j_c}\delta^c_{j_c}\left(\mathbf{x}^c_i,\mathbf{x}^c_l\right),
\end{equation}
where $p_n$ and $p_c$ are the number of numerical and categorical variables, respectively. 

This general formulation allows for a very flexible implementation of mixed distances. Gower's distance, for example, is a special case of this general distance. By definition, this general formulation satisfies the multivariate additivity property. Furthermore, although commensurability is not guaranteed, it can be achieved by applying appropriate variable  specific weights. In particular, if in (\ref{genmixeddist}) -- and analogously in (\ref{genmixeddist_formula}) -- we set
\begin{equation}\label{commensurableweights}
    w_j = 1/E[d_{j}(\bm{X}_i, \bm{X}_l)],
\end{equation}
for all $j$ and $d_j$, we obtain a distance that satisfies the multivariate additivity as well as the commensurability properties. Hence, using these weights we obtain an unbiased mixed distance where measurement types or units do not trivially impact the overall distance. 

The expected values in (\ref{commensurableweights}) depend on the functions $d_j$ and the distributions underlying the random variables. One may either consider appropriate assumptions concerning the underlying distributions, or follow a completely distribution-free approach to estimate the expected value based on the observed data. In Sections \ref{sect:numdist} and \ref{sect:catdist}, we consider some properties of common distance and scaling options for numerical and categorical variables. In particular, we show how, depending on the underlying distributions, several popular choices for $\delta_j$ and $w_j$ can lead to biased distances. 
 
\section{Numerical variables distances and scaling}\label{sect:numdist}
There are many distance definitions for numerical variables. Typically, these distances are based directly on the observed numerical differences. For example, consider the Minkowski distance, defined as:
\begin{equation*}
   d_{mink}\left(\mathbf{x}^n_{i},\mathbf{x}^n_{l}\right) = \left(\sum_{j_n=1}^{p_n} \lvert x^n_{ij_n}-x^n_{lj_n} \rvert^{1/m}\right)^{m},
\end{equation*} 
where $m \geq 1$. It is easily verified that this distance includes the Euclidean ($m=2$) and Manhattan distance ($m=1$) as special case. However, the multivariate additivity property is only satisfied for Manhattan distance. In this paper, we therefore only consider distances that apply Manhattan distances to the -- possibly transformed -- numerical data. That is, dropping the super- and sub-scripted $n$'s for ease of notation, we write 
\begin{equation*}
    d_{j}\left(\mathbf{x}_{i}, \mathbf{x}_{l}\right)=w_{j} \delta_{j}\left(\mathbf{x}_{i}, \mathbf{x}_{l}\right)
    = w_j \lvert f_j\left(\mathbf{x}_{i}\right) -f_j\left(\mathbf{x}_{l}\right) \rvert,
\end{equation*}
where $f_{j}(\cdot)$ is a function that transforms the original values (e.g., some type of scaling). In our multivariate additive setting, inserting this distance for all $p_n$ numerical variables gives, as total distance due to the numerical variables
\begin{equation}\label{gendistnumerical}
    d_{num}\left(\mathbf{x}_{i},\mathbf{x}_{l}\right) = \sum_{j=1}^{p_n} w_j \lvert f_j(\mathbf{x}_i)-f_j(\mathbf{x}_l) \rvert.
\end{equation}

Differences between observations for different variables are aggregated through the weights $w_j$ and functions $f_j(\cdot)$ that can also be used to account for differences in measurement scales and distributions. In multivariate contexts, common solutions for dealing with different measurement scales is to rescale the data using either the standard deviation, the range, or the robust (or interquartile) range. We refer to these options as independent scaling, as they do not involve the values of other variables. Alternatively, one could utilize the association between numerical variables to define distances. In the next section, we review common variants for scaling numerical variables and propose a new association-based scaling method that can be used in a multivariate additive context.

\subsection{Scaling of numerical variables}
\emph{Standard deviation scaling}. We define standard deviation scaling as the procedure that transforms, for each variable separately, the measurements to standard deviations from the mean. That is, 
\begin{equation*}
    f(\mathbf{x}_i)=f(x_{ij})=\frac{x_{ij}-\bar{x}_j}{s_j},
\end{equation*}
where $\bar{x}_j$ and $s_j$ denote, respectively, the sample mean and standard deviation for variable $j$.  This type of scaling converts all observations to so-called $z$-scores, and is a common procedure in many statistical analyses. Note that, regardless of the underlying distribution, by scaling in this way, we have: $E[X_i -X_l]^2=2$. Moreover, when the data are from a normal distribution, this conversion to $z$-scores converges to standard normal data.

\noindent \emph{Range scaling}. For range scaling, which is used in Gower's distance, we subtract the minimum and divide by the range of a variable as follows 
\begin{equation*}
      f(\mathbf{x}_i)=f(x_{ij})=\frac{x_{ij}-x^{min}_j}{x^{max}_j-x^{min}_j},
\end{equation*}
where $x^{min}_j$ and $x^{max}_j$ denote, respectively, the minimum and maximum for variable $j$. Consequently, all dissimilarities are between $0$ and $1$ and there is at least one dissimilarity equal to $1$. This scaling retains the original shape of the distribution, however, it is very sensitive to outliers and may lead to a loss of information. In particular, if the data are spread out over a large interval but highly condensed, the majority of distances becomes small, making differentiation difficult. 

\noindent \emph{Robust range scaling}. For robust range scaling, we subtract the median and divide by the interquartile range. That is,
\begin{equation*}
     f(\mathbf{x}_i)=f(x_{ij})=\frac{x_{ij}-x^{0.50}_j}{x^{0.75}_j-x^{0.25}_j},
\end{equation*}
where $x^{p}_j$ denotes the $p-$th percentile for variable $j$. This method overcomes some of the problems of range scaling since it is less susceptible to the presence of outliers.  

\noindent \emph{Principal component scaling}. The multivariate additivity property implicitly assumes independence of the variables. However, in some cases one could argue that relationships between the numerical variables should be taken into account. Our formulation accommodates this by allowing the functions $\delta^n_j$ to take the complete (numerical) observation vector as their argument. Here, we propose a method to incorporate variable association in a multivariate additive setting.

When two variables are (linearly) related, it could be argued that distances orthogonal to the direction of the relationship should receive larger weights than distances along the direction of the relationship. To achieve this, we first rotate the original data to their principal coordinates, and then standardize them using standard deviations.

Let $\mathbf{X}$ denote an $n \times p$ centered data matrix, of rank $p$, with $\mathbf{S}$ the corresponding sample covariance matrix. Consider the eigen-decomposition
\begin{equation*}
    \mathbf{S} = \mathbf{V} \bm{\Lambda}\mathbf{V}^{\top},
\end{equation*} where $\mathbf{V}$ is orthonormal and $\bm{\Lambda}$ is a $p \times p$ diagonal matrix of eigenvalues. Then 
\begin{equation*}
    \mathbf{X}^{*}=\mathbf{XV}\bm{\Lambda}^{-1/2} 
\end{equation*}
gives the data in standardized principal coordinates \parencite{Gower:1966}. The columns of $\mathbf{V}\bm{\Lambda}^{-1/2}$ define $p-$dimensional functions $f_j(\mathbf{x}_i)$ that can be implemented in our distance formulation (\ref{gendistnumerical}). Furthermore, if the eigenvalues are organized in non-increasing order, the first rotated variable, i.e., the first column of $\mathbf{X}^{*}$, corresponds to the linear combination of the original variables that captures the most variation. The second column accounts for the most of the remaining variance, and so on. Note that we do not reduce the number of variables, as is typically the case when applying principal component analysis, but only rotate the data. By considering the principal coordinate orientation, the dimensions are independent.

Principal component scaling is closely related to the Mahalanobis distance. However, with Mahalanobis distance, the standardized principal coordinates are rotated back to the original orientation. That is, 
\begin{equation*}
    \mathbf{X}^{*m}=\mathbf{XS}^{-1/2}=\mathbf{XV}\bm{\Lambda}^{-1/2}\mathbf{V}^{\top} = \mathbf{X}^{*}\mathbf{V}^{\top}. 
\end{equation*}
When Euclidean distances are used, this rotation does not affect the distances. Hence, Euclidean distances based on $\mathbf{X}^{*m}$ are equivalent to those based on $\mathbf{X}^{*}$. However, our multivariate additive distance is sensitive to rotation. We therefore retain the principal component orientation, $\mathbf{X}^{*}$. Consequently, in contrast to the other types of scaling considered, principal component scaling changes the orientation of the original variables, and our multivariate distance is additive with respect to the rotated rather than the original variables. 

\subsection{Commensurability of numerical variable distances} \label{sect:comnumdist}
The scaling options described above can be used to transform variables to similar scales. Nevertheless, this type of scaling does not necessarily lead to commensurability as defined in Definition \ref{DefCom} since the underlying distributions also impact the expected values. The commensurability property for the numerical variables can be expressed as follows 
\begin{equation}\label{EV=1}
    E[d_{j}\left(\bm{X}^n_i,\bm{X}^n_l\right)]= w_j E[\lvert f_j\left(\bm{X}^n_i\right)-f_j\left(\bm{X}^n_l\right) \rvert]=c,
\end{equation}
for all $i \neq l$ and $j \in \{1 \dots p_{n} \}$, and where $c$ is some constant that depends on the chosen scaling as well as on the distribution of $\bm{X}^{n}_{i}$. To illustrate the effect of different distributions and scaling options, we present simulation results for mean distances for various combinations of scaling and distributions in Table \ref{table:scaling_impact}. Moreover, in Appendix \ref{appendix:exactresultsU} and \ref{appendix:exactresultsN} we provide exact results concerning the mean values for Uniform distribution with range scaling, and for Normal distribution with standard deviation standardization. It can be verified that these exact results coincide with the results for $n=500$ in the Table \ref{table:scaling_impact}. 
\begin{table}[!ht]
\caption{Impact of distributions and scaling on mean distance $E[d_{j}(\bm{X}_i, \bm{X}_l)]$. Skewed refers to a $\chi^2_{1/2}$ distribution. For bimodal we considered $n/2$ draws from $\chi^2_{1/2}$  (censored at $10$), and $n/2$ draws from $10-\chi^2_{1/2}$ (censored at $0$)}.
\label{table:scaling_impact}
\centering
\begin{tabular}{|c|c|c|c|c|}
\hline
Distribution & Sample size & SD scaling & Range scaling & Rob. Range scaling\\ \hline
Normal & 50    & 1.16     & 0.25  & 0.76         \\ 
%Normal & 100       & 1.13     & 0.23  & 0.85         \\ 
Normal & 500    & 1.13     & 0.19  & 0.84         \\ \hline
Uniform & 50   & 1.16     & 0.35  & 0.69         \\ 
%Uniform  & 100      & 1.16     & 0.34  & 0.68         \\ 
Uniform & 500      & 1.15     & 0.33  & 0.67         \\ \hline
Skewed & 50   & 0.82     & 0.17  & 1.57     \\ 
%Skewed & 100   & 0.80     & 0.14  & 1.53     \\ 
Skewed & 500   &  0.77    & 0.09  & 1.49     \\ 
\hline
Bimodal & 50   & 1.07     & 0.50  & 0.51     \\ 
%Bimodal & 100   & 1.06     & 0.49  & 0.50    \\ 
Bimodal & 500   & 1.06     & 0.49  & 0.50     \\ 
\hline
\end{tabular}
\end{table}  

Table \ref{table:scaling_impact} illustrates the impact of the underlying distribution and the chosen scaling variant on the mean distances. In a multivariate setting, if all variables are scaled in the same way, and they are from the same distribution (possibly with different parameters; e.g., all Uniform or all Normal), the resulting values are on commensurable scales. However, when we have variables from different distributions, this is no longer the case. In particular, we see that for both range and robust range normalization the mean values vary considerably depending on the distribution. Standardization on the other hand, appears to be less sensitive to the underlying distributions. Still, even when standardization is used, distance calculations based on a combination of skewed and uniform or normal distributed variables clearly biases the results towards the uniform or normal variables. 

\section{Categorical variable distances and scaling}\label{sect:catdist}
Several distance definitions for categorical variable distance have been presented in the literature. Crucial for such distances is the quantification of differences between categories. Specifically, for two observations on one categorical variable, they either fall into the same category or different categories. In the case of differing categories, one straightforward approach is to assign a distance of $1$ (with no difference being assigned a value of zero). However, there may be reasons to differentiate distances depending on the different categories. For instance, the color green might be considered closer to blue than to red, necessitating a more nuanced distance measure. This challenge in quantifying distances for categorical data has led to a multitude of different proposals. For an overview of several such proposals, refer to, for example, \textcite{BCK:2008,SR:2019}.

In \textcite{vandeveldenetal2024}, a framework that allows a general and flexible implementation of categorical variable distances was introduced. Distances obtained using the framework are by construction multivariate additive. Here we use this framework to appraise properties of a selection of the categorical distances reviewed in \textcite{BCK:2008, vandeveldenetal2024}. Additionally, we consider some new distance variants that are based on the indicator matrix representation of categorical variables. 

\subsection{Category dissimilarity}\label{sect.catdissim}
Crucial for the definition of categorical variable distances is a so-called category dissimilarity matrix $\bm{\Delta}_j$ that quantifies, for each variable, the differences between categories. Simple matrix multiplication of an indicator matrix $\mathbf{Z}_j$ (i.e., a binary matrix where each column represents a category, and observed categories are indicated by ones) and $\bm{\Delta}_j$ results in the categorical variable distance matrix. That is, for variable $j$, distances between observations can be calculated as 
\begin{equation} \label{ZDeltaZ'}
    \mathbf{D}_{j} = \mathbf{Z}_{j} \bm{\Delta}_j \mathbf{Z}^{\top}_{j} 
\end{equation}

In a multivariate setting, the distance simply becomes the sum of the variable specific distances over the variables. Hence, the distances satisfy the multivariate additivity property. From (\ref{ZDeltaZ'}) it follows that different $\bm{\Delta}_j$'s correspond to different definitions of categorical variable distances. 

In \textcite{vandeveldenetal2024}, a distinction is made between \say{independent} and \say{association-based} category dissimilarities. For independent category dissimilarities, one variable is treated at the time. The distribution over the categories for that variable can be taken into account to construct the dissimilarities but any possible association between the different variables is ignored. However, to better differentiate between categories, associations between categorical variables may be employed as well. The resulting distances can be referred to as association-based. 

\subsubsection{Independent category dissimilarities} \label{sect.independencategory}
Table \ref{table:deltas} provides an overview of category dissimilarity matrices $\bm{\Delta}_j$. The independent dissimilarity matrices in the top panel are a subset of the distances reviewed in \textcite{BCK:2008} whereas those in the bottom panel, which we label as \say{indicator} variants, concern new proposals based on the application of numerical variable distances to possibly scaled, indicator matrices. 

\begin{table}[!ht]
\caption{Distances and corresponding category dissimilarity matrices or category dissimilarites $\delta_{ab}$, where $\odot$ indicates the Hadamard (i.e., elementwise) matrix product and $\log(\cdot)$ takes the logarithm of the parenthesized object and collects them in an object of the same size}
\label{table:deltas}
\centering
\begin{tabular}{|l|l|}
\hline
Distance & Category disimilarity matrix $\bm{\Delta}_j$ (or \\ & its typical element $\delta_{ab}$, for $a \neq b$)  \\ \hline
Matching  & $\bm{\Delta}_{m_j} = \mathbf{1} \mathbf{1}^{\top} - \mathbf{I}$  \\ 
Eskin & $\bm{\Delta}_{e_{j}} = 2/q_j^2\bm{\Delta}_{m_j}$ \\ 
Occurence frequency (OF) & $\bm{\Delta}_{OF_{j}} = \log(\mathbf{p}_j)\log(\mathbf{p}_j)^{\top} \odot \bm{\Delta}_{m_j}$  \\ 
Inverse OF & $\bm{\Delta}_{IOF_{j}} = \log(n\mathbf{p}_j)\log(n\mathbf{p}_j)^{\top} \odot \bm{\Delta}_{m_j}$ \\
\hline
Indicator: No scaling & $\bm{\Delta}_{d_{j}}=2\bm{\Delta}_{m_j}$  \\
Indicator: Hennig-Liao scaling & $\bm{\Delta}_{HL_{j}} = 2\eta_{j}\bm{\Delta}_{m_j}$ \\
Indicator: Standard deviation scaling & $\delta^{s}_{ab_{j}}=\sqrt{\frac{1}{q_j}} \left(s^{-1/2}_a + s^{-1/2}_b \right)$ \\
Indicator: Cat. dissimilarity scaling & $\bm{\Delta}_{cds_{j}}=\frac{1}{q_j}\mathbf{S}^{-1/2}_{j_d}\bm{\Delta}_{m_j} \mathbf{S}^{-1/2}_{j_d} $\\
%Indicator: MCA scaling & $\bm{\Delta}_{mca_{j}}=\frac{1}{q_j} \mathbf{D}^{-1/2}_{p_j} \bm{\Delta}_{m_j}\mathbf{D}^{-1/2}_{p_j} $\\
\hline
\end{tabular}
\end{table}

The idea of the indicator-based category dissimilarities is to treat the binary coded data as \say{numerical}. In practice, this type of coding, in combination with Euclidean distance, appears to be not uncommon, \parencite[see, e.g.,][]{BG:2020, HennigLiao2013}. As Euclidean distance is not multivariate additive, we replace it by the Manhattan distance. Distances resulting from applying Manhattan distances to differently scaled indicator matrices can then be implemented using (\ref{ZDeltaZ'}) by appropriately defining the category dissimilarity matrices. Below we briefly show how different choices for scaling lead to the different indicator-based category dissimilarity matrices presented in Table \ref{table:deltas}.\\ 

\noindent \emph{Indicator: No scaling}
Applying Manhattan distance to a one-variable indicator matrix results in distances of either zero (same category was chosen) or $2$, different categories are chosen. Hence, in a multivariate setting one can implement this distance by defining, for each variable $j$ as category dissimilarity matrix $\bm{\Delta}_j=2\bm{\Delta}_{m_j}$, where $\bm{\Delta}_{m_j}$ is the matching category dissimilarity matrix of appropriate order.\\ 

\noindent \emph{Indicator: Hennig and Liao scaling}
\textcite{HennigLiao2013} propose to tackle the differences caused by different variable types by standardizing variances across variables. 
For a categorical variable with $q_j$ categories, coded using a $q_j-$dimensional indicator matrix, it is proposed to scale the indicator matrix in such a way that
\begin{equation}\label{HLscaling}
    \sum_{a \in \{1 \dots q_j\}} {E[Z^{*}_{ia} - Z^{*}_{la}]^2}  = 2\phi,
\end{equation} 
where $Z^{*}_{ia}$ denotes the scaled indicator variable for the $i-$th observation on the $a-$th category, and $\phi$ needs to be supplied by the researcher. \textcite{HennigLiao2013} suggest to use $\phi=1/2$ for categorical variables and $\phi=(q_j-1)/q_j$ for ordinal ones. Consequently, for the ordinal case, the expected value in (\ref{HLscaling}) converges to $2$, similar to the situation for numerical variables. 

To satisfy (\ref{HLscaling}), a scaling factor, say $\eta_j$, can be applied to each indicator matrix. That is, for the $j-$th categorical variable, $\mathbf{Z}^{*}_j = \eta_j\mathbf{Z}_j$. Distances are then calculated between the rows of the scaled indicator matrix. The Hennig and Liao scaling factor $\eta_j$ is implemented in the function \code{distancefactor} in the \proglang{R} package \pkg{fpc} \parencite{fpc_pkg}. In the Appendix we show how this $\eta_j$ is calculated and derive limiting results.  

\textcite{HennigLiao2013} use Euclidean distance between combined standardized numerical and rescaled indicator matrices, but this does not meet the multivariate additivity constraint. Switching to Manhattan distance solves this issue. For categorical variables, the Manhattan distance between rows of rescaled indicator matrices is either $0$ if the categories match, or $2\eta_j$ if they differ. Thus, Hennig and Liao's scaling with Manhattan distance can be integrated into our mixed variable distance by constructing category dissimilarity matrices for each variable $j$ as follows 
\begin{align}
    \bm{\Delta}_{HL_j} = 2\eta_j\bm{\Delta}_{m_j},
\end{align}
where $\bm{\Delta}_{m_j}$ denotes the matching category dissimilarity matrix of appropriate order.\\

\noindent \emph{Indicator: standard deviation scaling}
By considering the columns of each indicator matrix, representing the categories of the categorical variable, as numerical variables that, together, represent one categorical variable, we can also apply the standard deviation scaling procedure for numerical variables to these indicator matrices. The sample covariance matrix corresponding to the indicator matrix $\mathbf{Z}_j$ corresponding to a categorical variable with $q_j$ categories, can be defined as 
\begin{align} \label{Sd}
    \mathbf{S}_{j}= \frac{1}{n}\mathbf{Z}_{j}^{\top}\left( \mathbf{I}-(1/n)\mathbf{11}^{\top} \right) \mathbf{Z}_{j}= \mathbf{D}_{p_{j}} -\mathbf{p}_{j}\mathbf{p}_{j}^{\top},
\end{align}
where $\mathbf{p}_{j}$ denotes the vector of observed proportions and $\mathbf{D}_{p_{j}}$ is a diagonal matrix with $\mathbf{p}_{j}$ as its diagonal. We define a scaled indicator matrix as
\begin{align} \label{StandZ}
    \mathbf{Z}^{*}_{j} = \frac{1}{\sqrt{q_j}}\mathbf{Z}_j \mathbf{S}^{-1/2}_{j_d},
\end{align}
where $\mathbf{S}_{j_d}$ is a diagonal matrix with on its diagonal the diagonal elements of $\mathbf{S}_j$. Note that $\mathbf{Z}^{*}_j$ is scaled such that the variance per category equals $1/q_j$. Consequently, the variance per variable is standardized to $1$. 

Standardizing all variables in this fashion, arranging them next to each other and considering the Manhattan distance between the rows of the resulting standardized \say{super}-indicator matrix, corresponds to distances obtained by applying (\ref{ZDeltaZ'}), where the category dissimilarity matrices $\bm{\Delta}_j$ have as its elements the absolute values of the differences between the rows of the standardized indicator matrix $\mathbf{Z}^{*}_{j}$ corresponding to different categories of variable $j$. That is,
\begin{align*}
    \delta_{ab_{j}}=\frac{1}{\sqrt{q_j}} \left(s^{-1/2}_a + s^{-1/2}_b\right),
\end{align*}
where, $a,b \in {1,\dots q_j}$ denote categories of variable $j$, and, generically, $s_a = p_{j_a}(1-p_{j_a})$, denotes the $a-$th diagonal element of $\mathbf{S}_{j_d}$.\\

\noindent \emph{Indicator: Category dissimilarity scaling}
An alternative way that combines the categorical distance framework and scaling of the indicator matrix, is to directly substitute the scaled indicator matrix into (\ref{ZDeltaZ'}) and apply the scaling to the category dissimilarity matrix rather than to the indicator matrices. Using standard deviation scaling and simple matching, we get
\begin{align*}
    \mathbf{D}_{j} = \mathbf{Z}_j^* \bm{\Delta}_{m_j}\mathbf{Z}_j^{*\top} = \frac{1}{q_j} \mathbf{Z}_j \mathbf{S}^{-1/2}_{j_d} \bm{\Delta}_{m_j} \mathbf{S}^{-1/2}_{j_d} \mathbf{Z}^{\top} = \mathbf{Z}_j \bm{\Delta}_{cds_j} \mathbf{Z}_j^{\top},
\end{align*}
where
\begin{equation}\label{Delta_s}
    \bm{\Delta}_{cds_j}=\frac{1}{q_j} \mathbf{S}^{-1/2}_{j_d} \bm{\Delta}_{m_j}\mathbf{S}^{-1/2}_{j_d}.
\end{equation}

Using this definition, individual pair-wise category dissimilarities become
\begin{align*}
    \delta_{ab}= \frac{1}{q_j}\left[s_{a}s_{b})\right]^{-1/2}. 
\end{align*}

\subsubsection{Association-based category dissimilarity}\label{sect:catdist_ab}
The categorical dissimilarities in Table \ref{table:deltas} define category dissimilarities for each variable independent of the other variables. Alternatively, several authors,  \parencite[e.g.,][]{LH:2005,AD:2007,JCL:2014,ROBNLH:2015} have proposed distance measures for categorical variables that take into account the association between categorical variables. Crucial in the construction of such \say{association-based} category dissimilarities are the joint bivariate distributions. In \textcite{vandeveldenetal2024}, it was shown that extant association-based distances correspond to different ways of quantifying the differences between the conditional distributions. That is, the distributions for categories of one variable over the categories of another variable. 

Let $\mathbf{P}^{jk}$ denote the (empirical) joint distribution for categorical variables $j$ and $k$. The conditional distribution for the row variable $j$ can be obtained by dividing the elements of $\mathbf{P}^{jk}$ by the marginal distribution for the column variable $k$. Using the indicator matrices as introduced before, we can derive the conditional distribution for the categories of variable $j$ over the categories of variable $k$ as  
\begin{equation*}
\mathbf{R}^{j|k}=\left(\mathbf{Z}^{\top}_j\mathbf{Z}_j\right)^{-1}\mathbf{Z}_j^{\top}\mathbf{Z}_k.
\end{equation*}
We can define the category dissimilarities for the categories of variable $i$, by quantifying the differences between the rows of $\mathbf{R}^{j|k}$. 

For illustration purposes, we briefly review, in a bi-variate setting, two variants: 1) Kullback-Leibler variant \parencite{LH:2005} and 2) Total variance distance \parencite{AD:2007, vandeveldenetal2024}. \\

\noindent \emph{Kullback-Leibler}
\textcite{LH:2005} propose to use a symmetric variant of Kullback-Leibler divergence to define the category dissimilarities. That is, the category dissimilarity for categories $a$ and $b$ of variable $j$, is defined as the sum of the Kullback-Leibler divergences between rows $\mathbf{r}_{a}$ and $\mathbf{r}_{b}$:
\begin{align*}
    \delta_{ab}=&\text{KL}\left(\mathbf{r}_{a},\mathbf{r}_{b}\right) +\text{KL}\left(\mathbf{r}_{b},\mathbf{r}_{a}\right)\\
    =& \sum_{l=1}^{q_k} \left[ r_{al}\log_2\left(\frac{r_{al}}{r_{bl}} \right) + r_{bl}\log_2\left( \frac{r_{bl}}{r_{al}} \right) \right]
\end{align*}
where we dropped the superscribed $j|k$'s, $\text{KL}$ denotes the Kullback-Leibler divergence between two discrete distributions and $\log_2$ denotes the base 2 logarithm.\\

\noindent \emph{Total variance distance}
\textcite{AD:2007} proposed a mixed variable distance that combined Euclidean distances for numerical variables with a proposal for the categorical variable distances that, as shown by \textcite{vandeveldenetal2024}, is equivalent to $1/2$ times the Manhattan distances between rows of $\mathbf{R}$:
\begin{equation*}
    \delta_{ab}=1/2 \|\mathbf{r}_{a} - \mathbf{r}_{b}\|_1,
\end{equation*}
where $\|\cdot\|_1$ denotes the $L_1$ norm. This is also known as the total variance distance between two discrete probability distributions. \\

\noindent Note that for these association-based measures, the category dissimilarities become zero if the two variables are independent. For a dependent scenario where the marginal probabilities coincide (i.e., we have the same number of categories and the joint distribution can be represented by a diagonal matrix), the category dissimilarity values using total variance distance is equivalent to the matching case: $\bm{\Delta}_m$. For Kullback-Leibler, perfect dependence leads to category dissimilarity matrix where all non-diagonal values require assessment of $r_{al}log_2\left(\frac{r_{al}}{r_{bl}}\right)$ where $r_{al}=1$ and $r_{bl}=0$. We can resolve this by replacing the zero conditional probability by a small value. The Kullback-Leibler implementation in the \proglang{R} package \pkg{philentropy} uses $r_{bl}=\left(1/10\right)^5$.  

In a multivariate context, bivariate associations need to be aggregated. For instance, an average category dissimilarity can be calculated by considering all pairs of categorical variables or only pairs with associations above a certain threshold may be considered.  \textcite{vandeveldenetal2024} describe several association-based methods, with more available in an \proglang{R} package on GitHub at \url{https://github.com/alfonsoIodiceDE/catdist_package}. It is important to note that different methods produce varying category dissimilarities, so the choice of method affects the overall distances. Care is needed when combining different variants.

\subsection{Commensurability of categorical variable distances} \label{comcatdist}
The general framework for constructing categorical variable distances allows for a straightforward assessment of commensurability. In particular, let $X_{ij}$ denote the $i-$th observation for categorical variable $X_{j}$. Let $q_j$ denote the number of categories for this $j-$th variable and let $\bm{p}_j$ denote the $q_j \times 1$ vector of category probabilities. Using (\ref{ZDeltaZ'}) as distance function, the mean absolute distances can be expressed as follows
\begin{equation}\label{EVCat}
    E[d_j({\bm X}_{i}, {\bm X}_{l})]= \sum_{a,b \in \{1 \dots q_j\}} {p_{j_a}p_{j_b}\delta_{j_{ab}}} = \bm{p}^{\top}_j\bm{\Delta}_j \bm{p}_j,
\end{equation}
where $\bm{\Delta}_{j}$ is the matrix of category dissimilarities for variable $j$ whose elements, $\delta_{j_{ab}}$ quantify the dissimilarity between categories $a$ and $b$. Hence, the mean distance is a function of variable specific properties (the number of categories), distributional aspects ($\bm{p}_j$) and the choice of distance (the category dissimilarity matrix $\bm{\Delta}_j$). We illustrate this below by considering different combinations of distances and distributions. For instance, Table \ref{table:expectedvalues} shows category dissimilarities and expected values for one categorical variable where the distribution over categories is uniform. Moreover, for the association-based variants, the category dissimilarity matrices are presented corresponding to perfect bivariate dependence where the two variables have the same number of categories.

\begin{table}[!ht]
\caption{Distances and corresponding category dissimilarity matrices when category probabilities are equal. That is, $p_j=1/q$ for $j=1 \dots q$. For Hennig-Liao we used $\phi=1/2$ and the limiting result ($n \rightarrow \infty$) for $\eta$. For Inverse OF we used $n=160$. For Kullback-Leibler, $\kappa=5\log_2(10) \approx 16.610$. We dropped all subscripted $j$'s for ease of notation}
\label{table:expectedvalues}
\centering
\begin{tabular}{|l|l|c|c|c|}
\hline
Distance & Cat. dissimilarity $\bm{\Delta}_j$ & $E[d(X_{i}, X_{l})] $&  $ q=2$  & $ q=5  $\\ \hline
Matching  & $\bm{\Delta}_m = \mathbf{1} \mathbf{1}^{\top} - \mathbf{I}$  & $\frac{q-1}{q}$ & $0.5$ & $0.8$\\ 
Eskin & $\bm{\Delta}_e = 2/q^2\bm{\Delta}_m$ & $\frac{2(q-1)}{q^3}$ & $0.250$ & $0.064$ \\ 
Occurence frequency (OF) & $\bm{\Delta}_{OF} = \log^2(q)\bm{\Delta}_m$ & $\log^2(q)\frac{q-1}{q}$ & $0.240$ &  $2.072$\\ 
Inverse OF & $\bm{\Delta}_{IOF} = \log^2(n/q) \bm{\Delta}_m$ & $\log^2(n/q)\frac{q-1}{q}$ & $9.601$ & $9.610$ \\ \hline
Indicator (no scaling) & $\bm{\Delta}_{d}=2\bm{\Delta}_m$ & $\frac{2(q-1)}{q}$ &$1$ & $1.6$\\
Indicator (Hennig-Liao scaling) & $\bm{\Delta}_{HL} = \sqrt{\frac{2q}{q-1}}\bm{\Delta}_m$ & $\sqrt{\frac{2\left(q-1\right)}{q}}$ &$1$ & $1.265$\\
Indicator (St. dev. scaling) & $\bm{\Delta}_{s}=2\sqrt{\frac{q}{q-1}}\bm{\Delta}_m$ &  $2\sqrt{\frac{(q-1)}{q}}$ & $1.414$ & $1.789$\\
Indicator (Cat. dissim. scaling) & $\bm{\Delta}_{cds}=\frac{q}{q-1}\bm{\Delta}_m $ & $1$ &$1$&$1$\\ \hline
%Indicator (MCA scaling) & $\bm{\Delta}_{mca}=\bm{\Delta}_m $ & $\frac{q-1}{q}$ & $0.5$ & $0.8$\\ \hline
Total Variance  & $\bm{\Delta}_{tvd} = \bm{\Delta}_m$  & $\frac{q-1}{q}$ & $0.5$ & $0.8$\\ 
Kullback-Leibler (Le \& Ho) & $\bm{\Delta}_{KL} = \kappa\bm{\Delta}_m$ & $\kappa\frac{q-1}{q}$ & $8.305$ & $13.288$ \\ 
%Chi squared & $\bm{\Delta}_{\chi} = 2q\bm{\Delta}_m$ & $2(q-1)$ & $2$ &  $8$\\ 
%Joint Entropy & $\bm{\Delta}_{JE} = \frac{2\log_2(q)}{q}\bm{\Delta}_m$ & $\frac{2(q-1)\log_2(q)}{q^2}$ & $0.5$ & $0.743$ \\
\hline
\end{tabular}
\end{table}

It is important to note that for uniformly distributed variables, all category dissimilarity matrices in Table \ref{table:expectedvalues} can be expressed as a function of the simple matching variant, that is, all category dissimilarities are constant. Furthermore, we see that, except for the scaled category dissimilarity matrix, the number of categories influences the mean absolute distances. Consequently, when categorical variables with different numbers of categories are used, the overall distance is directly influenced by the number of categories. For non-uniform distributions, analytical results can also be obtained using (\ref{EVCat}) if a distribution is provided. In Appendix \ref{appendix:skewedistributions} we illustrate the effect of different distributions on the distances in Table \ref{table:deltas} by considering, for different numbers of categories, different stylized distributions exhibiting various degrees of skewness.  

\section{Unbiased mixed variable distances} \label{sect:unbiased}
In Sections \ref{sect:numdist} and \ref{sect:catdist}, we showed that depending on distributions, scaling and choice of distance function, multivariate distances typically do not satisfy the commensurability property even in a single variable type setting. Furthermore, from the values in Tables \ref{table:scaling_impact} and \ref{table:expectedvalues} it is clear that combining numerical and categorical variable distances leads to variable type dependent biases. 

To resolve these issues, we can apply weights based on the mean distance as defined in Equation (\ref{commensurableweights}). For numerical data, one can estimate the mean distances based on the observed data. That is, for given $f_j$, we estimate the mean distances as follows  
\begin{equation*}
    m_{j_n} = \frac{1}{n} \sum_{i \neq l}^{n} {d\left(\mathbf{x}^n_{i},\mathbf{x}^n_{l}\right) } = \frac{1}{n} \sum_{i \neq l}^{n} {\lvert f_j(\mathbf{x}^n_i)-f_j(\mathbf{x}^n_l) \rvert}.
\end{equation*}
For $n$ sufficiently large, this estimate converges to the population mean. Hence, by defining weights as
\begin{equation}\label{weightsnum}
w_{j_n}=1/m_{j_n}
\end{equation}
we obtain commensurability. We can therefore interpret the scaled measurement units of the variables as \say{mean distance} units.  

Table \ref{table:expectedvalues} along with Figures \ref{fig:expvals1} through \ref{fig:exvals_assbased} in Appendix \ref{appendix:skewedistributions}, show that for categorical variables the impact of the number of categories, underlying distributions and category dissimilarity definitions on mean distances can be significant. Consequently, in a multivariate setting, choosing the same category variable distance for all variables may not lead to commensurability. Similar to the case for numerical variables, we can resolve this by defining weights based on the mean distance. In particular, for a selected distance (i.e., $\bm{\Delta}_{j}$) and known distributions $\bm{p}_j$, we achieve commensurability by using 
\begin{equation}\label{weightscat}
    w_{j_c} = 1/\left(\bm{p}^{\top}_j\bm{\Delta}_j \bm{p}_j\right).
\end{equation}
If, as is typically the case, $\bm{p}_j$ is unknown, one can either use estimates based on theory, or estimate $\bm{p}_j$ from the data. Note that through the different definitions of $\bm{\Delta}_j$, it is still possible to consider different, yet commensurable, distances. 

Inserting the weights (\ref{weightsnum}) and (\ref{weightscat}) into the general mixed variable distance (\ref{genmixeddist_formula}) results in a family of unbiased distances. That is, regardless of the choice of distance for the numerical and categorical variables, the resulting distance is multivariate additive and commensurable. 

\section{Simulation study}\label{sect:simulation}
As shown in Sections \ref{sect:numdist} and \ref{sect:catdist}, individual variables may have varying impacts on overall distance, with observed differences being more or less influential depending on the data distribution, scaling method, or data type. 

To illustrate the impact of different choices, we conduct a simulation study generating high-dimensional numerical data based on an underlying low-dimensional structure. We then discretize a subset of the variables to create mixed data. Next, we analyze how various mixed-distances and levels of discretization (i.e., number of categories) affect the results.

For the distances, we consider the following variants: 1) \textit{Numerical data} Manhattan distance on the original generated data (without discretizing any variables) 
2) \textit{Naive approach}: Euclidean distance on the complete standardized matrix where categorical variables are transformed to indicator matrices (one-hot-encoding); 3) \textit{Hennig-Liao}: Euclidean distance on the matrix with numerical variables standardized and categorical variables transformed to indicator matrices (one-hot-encoding) and scaled using the Hennig-Liao scaling factor; 4) \textit{Additive Hennig-Liao}: Manhattan distance with numerical variables standardized and Hennig-Liao scaling of the indicator matrices (i.e., $\bm{\Delta}_{HL}$); 5) \textit{Gower}: range normalized numerical variables and simple matching for categorical variables; 6) \textit{Unbiased independent}: commensurable distance using simple matching for the categorical variables. 7) \textit{Unbiased standardized}: commensurable distance using category dissimilarity scaling for the categorical variables; 8) \textit{Unbiased dependent}: commensurable association-based mixed distance using PCA scaling of the numerical, and total variance distance for the categorical variables. Note that the variants 2) and 3) are not multivariate additive. We include them here as benchmark of methods used in practice. 

\subsection{Data generating process}
For each simulation instance, we first generate a $2-$dimensional configuration, $\mathbf{Y}$, by sampling $n \times 2$ times from a uniform distribution on $[-2,2]$ and constructing an orthogonal basis. Next, we post-multiply $\mathbf{Y}$ by a random $2 \times p$ matrix, $\mathbf{N}$, with elements drawn from a uniform distribution on $[-2, 2]$. The resulting $n \times p$ matrix, $\mathbf{X}_o = \mathbf{YN}$, gives a numerical, $p-$dimensional \say{observed} matrix corresponding to the $2-$dimensional configuration $\mathbf{Y}$. We add random noise to these observations using a normal distribution with $\sigma = 0.03$ (approximately half the standard deviation of the generated values). Finally, we discretize $p_{cat}$ columns of $\mathbf{X}_o$ by splitting their ranges into the desired number of intervals. This setting allows us to appraise the effect of the type of variable \emph{within} a chosen mixed-distance variant, and with respect to retrieval of the underlying configuration. 

\subsection{Variable specific effects}\label{sect:varspecific}
We generate $100$ instances, each with $n=500$ observations on $p=6$ variables. For the mixed datasets, two variables remain numerical, the other $4$ variables are converted into $2, 3, 5$ and $9$ categories, respectively. To assess the effect of individual variables we consider the effect of leaving one-variable-out on the distance matrix. In particular, for each instance, we consider the absolute values of the differences between the leave-one-out distance and the distance using all variables. Furthermore, as the different mixed distance variants are on different scales, we also calculate a relative measure by dividing the absolute difference for each variable, by the sum of these differences over all variables. 

The results are summarized in Figure \ref{fig:varimportances_mad}. The left panel shows that, for the additive variants, the observed values align with the theoretical ones in Tables \ref{table:scaling_impact} and \ref{table:expectedvalues}. For example, for the numerical variables, we see that the mean values for the non-commensurable additive variants (Hennig-Liao, numerical and Gower\footnote{The Gower implementation in \code{daisy} adds an additional scaling factor by dividing Gower distance by the number of variables. Here we did not apply this scaling}) are between the corresponding values in Table \ref{table:scaling_impact} for the normal and uniform distributions. (The generated data are a mixture of uniform and normal). For the categorical variables we see a similar correspondence to the values in Table \ref{table:expectedvalues}. For the Euclidean variants, we did not derive theoretical mean values for individual variables as these depend on the number of variables. In general, for Euclidean distance, individual mean values decrease as the number of variables increases. We can see this effect by comparing the two Hennig-Liao implementations. For both variants, having more categories increases the distances due to that variable. 

For the two independent unbiased variants (\textit{Uind} and \textit{Ustd}) the values are, by construction, $1$. For the dependent variant, however, leaving a variable out also impacts the distances of the other categorical variables. Hence, for this variant the differences compared to the full distance are close to, but not exactly, $1$.

To assess the contributions of variables of different types to the overall distance, we consider the relative variable contributions depicted in the right panel of Figure \ref{fig:varimportances_mad}. We see that for the numerical data, as well as for the unbiased variants, the relative effect is either exactly $1/6$ or close to this value. For the other variants, distances are affected by type (i.e., numerical or categorical) as well as by the number of categories. For \textit{Naive}, and, to a lesser extent, \textit{Gower}, categorical variables contribute more to the overall distance. \textit{Hennig-Liao} reduces the influence of categorical variables, increasing the relative contribution of numerical variables to the distance. 

\begin{figure}[!ht]
\center
\includegraphics[scale=.7]{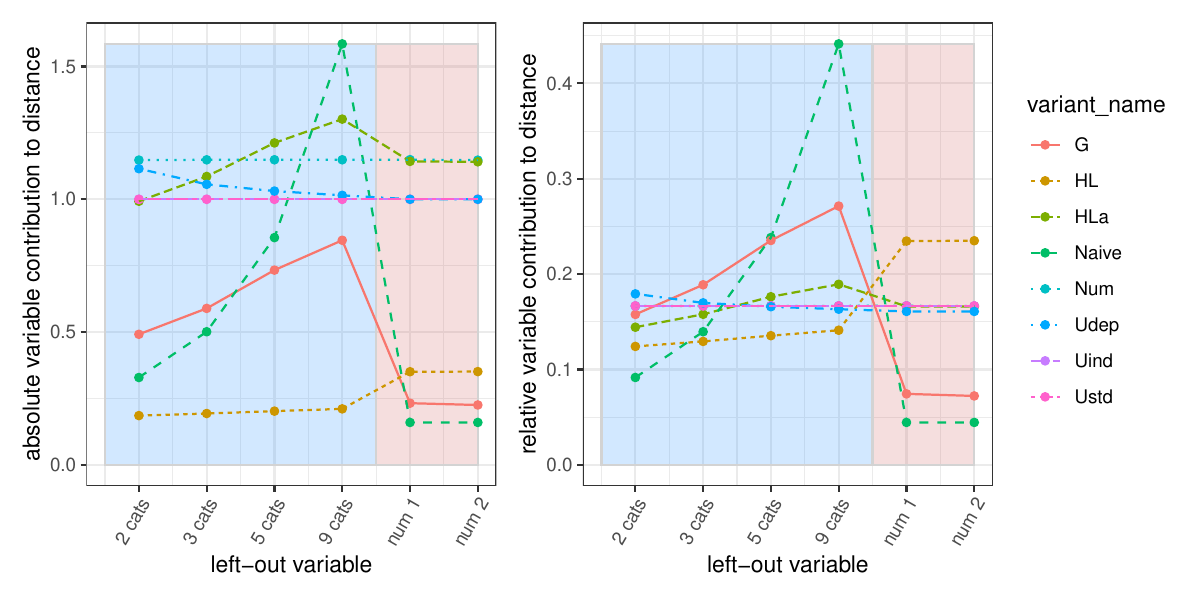}
\caption{Effect of leaving variables out on distance. The first 4 variables (blue background) are categorical with indicated numbers of categories respectively. The last two are numerical. For the variants: G = Gower, HL = Hennig-Liao, HLa = Additive Hennig-Liao, Naive = Naive approach, Num = original data (all variables numerical), Udep = Unbiased dependent, Uind = Unbiased independent, Ustd = Unbiased standardized}
\label{fig:varimportances_mad}
\end{figure}

To see how these differences impact subsequent analyses, we consider the effect of individual variables on a $2-$dimensional classical multidimensional scaling (MDS) solution. That is, for all distance variants, we compare the MDS solutions when leaving one-variable-out, with the solution using all variables. To measure the difference, we consider the alienation coefficient \parencite{BL:1985}, which lies between zero and one and can be seen as a measure of unexplained variance. Hence, higher values indicate more influence on the solution. The results are visualized in Figure \ref{fig:alienations}. Note that the effects on mean distances do not carry-over one-to-one to the effect on the low-dimensional configurations. In particular, although the categorical variables impact the MDS solution more than the numerical ones for Gower and Naive, the effect of the number of categories is reversed. That is, variables with fewer categories have more impact on the MDS solutions. In general, low category variables allow for little differentiation between observations and projecting to a low-dimensional space suffers from this for all distances. 

\begin{figure}[!ht]
\center
    \includegraphics[scale=.7]{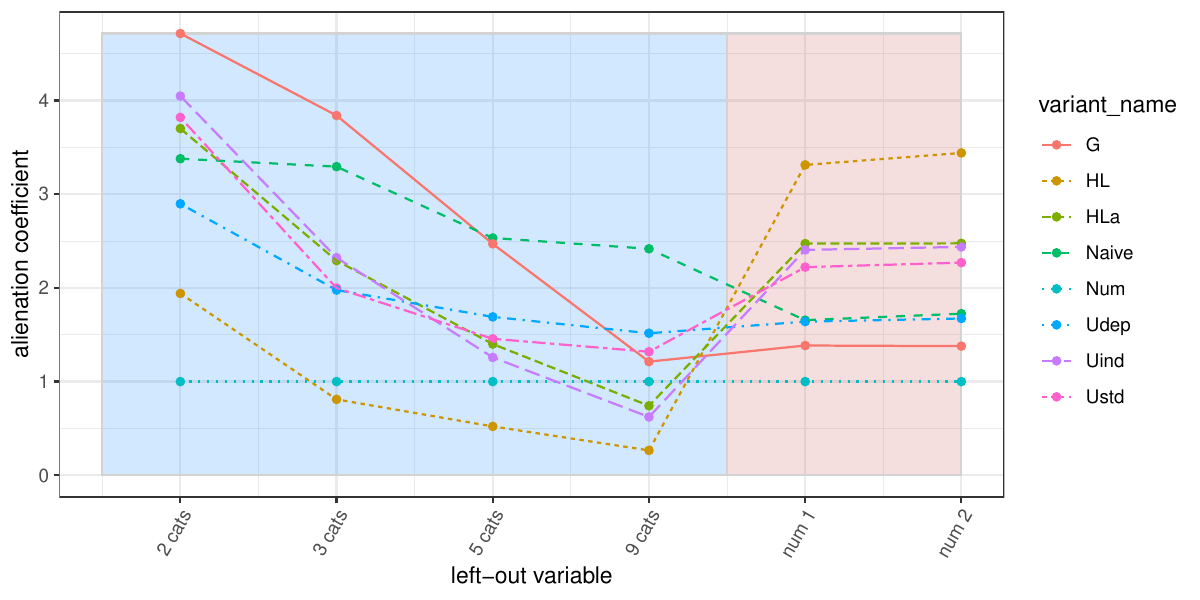}
\caption{Effect of leaving variables out on the MDS solution. The first 4 variables (blue background) are categorical with indicated numbers of categories respectively. The last two are numerical. Variant abbreviations same as in Fig \ref{fig:varimportances_mad}}
\label{fig:alienations}
\end{figure}

\subsection{Retrieval of the true configuration}
To assess the effect of the different mixed-distances on retrieving the underlying configuration, we generate $100$ instances, each with $n=500$ observations on $p=6$ variables. Three variables remain numerical, and three variables are all discretized into the same number of categories, i.e., either $2, 3, 5$ or $9$. For each dataset, we create a distance matrix using the same mixed distances as before. Next, we apply classical MDS and calculate the alienation coefficient of this solution with the true $2-$dimensional configuration. The results are presented in Figure \ref{fig:alienation_boxplots}.

We see that different distances perform similarly when only binary variables are added. However, as the number of categories increases, the variation in results (as indicated by the sizes of the boxes in Figure \ref{fig:alienation_boxplots}) gets larger. Results for the unbiased variants appear to improve when more categories are used. \textit{Unbiased dependent} performs better when dealing with more than $2$ categories, as it effectively distinguishes categories based on their associations. Recall that the high-dimensional data ($p=6$) were generated by expanding $2-$dimensional data, therefore, meaningful associations are present. The association-based approach utilizes these dependencies.

The results in this plot are in line with the findings of the variable specific effects in Figure \ref{fig:varimportances_mad}. For example, comparing Figures \ref{fig:varimportances_mad} and \ref{fig:alienations} we see how distances using \textit{Hennig-Liao} over-emphasize the numerical variables. Similarly, by using range normalization, distance due to the numerical variables is undervalued in \textit{Gower} distance. Consequently, when the true configuration relies more on numerical variables, \textit{Gower} distance performs worse. 

\begin{figure}[!ht]
\center
\includegraphics[scale=.65]{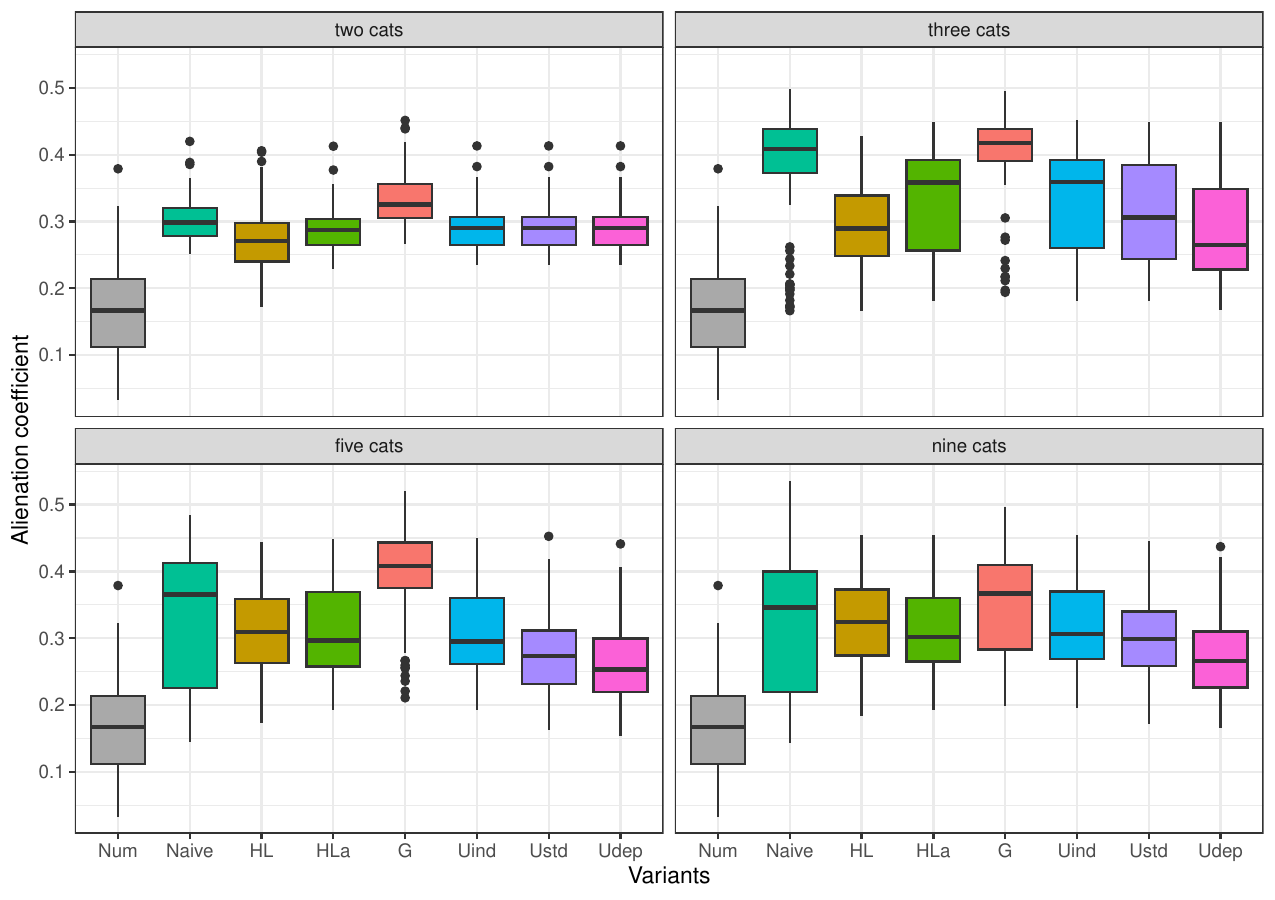}
\caption{Distribution of alienation coefficients (with respect to the true configuration) of $2-$dimensional classical MDS solutions using different mixed distances, where $3$ out of $6$ numerical variables have been discretized into the indicated number of categories}
\label{fig:alienation_boxplots}
\end{figure}

\section{Illustration: FIFA player data}\label{sect:illustration}
The simulation study confirmed, in a controlled setting, the theoretical findings of Sections \ref{sect:numdist} and \ref{sect:catdist}. To further illustrate these findings, we apply a selection of the mixed variable distances to an empirical dataset. In particular, we consider a sample from the FIFA player data 2021 that can be found at \href{https://www.kaggle.com/datasets/stefanoleone992/fifa-21-complete-player-dataset}{kaggle.com}\footnote{The complete link is \url{https://www.kaggle.com/datasets/stefanoleone992/fifa-21-complete-player-dataset}}. In particular, we select $408$ players from the Dutch league. Furthermore, we select $7$ categorical and $7$ numerical variables. Our selection of variables is based on our aim to illustrate how different mixed distances are affected by different measurement types and distributions. The categorical variables have categories ranging from $2$ (\say{preferred foot}) to $25$ (\say{position}). With the exception of the variable \say{club}, the distributions over categories are typically uneven. For the numerical variables, the wage and, particularly, transfer value variables are skewed. See Figure \ref{fig:fifa_desc} in Appendix \ref{appendix:FIFA} for an overview of the distributions of the variables. 

For this dataset, there is no known underlying \say{truth} concerning the distances or subsequent distance-based analysis. We can, however, illustrate how different choices for distances lead to trivial, scale and measurement type-related, biases. In fact, based on the results presented in Sections \ref{sect:numdist} and \ref{sect:catdist}, we expect that for all non-commensurable variants, distances due to the categorical variables increase with the number of categories. Furthermore, we expect that for Gower, the distance due to categorical variables exceeds that of the numerical variables. The Hennig-Liao scaling is designed to remove this bias towards categorical variables.

Similar to our analysis of simulated data, we first consider the effects of leaving individual variables out on the overall distance. Figure \ref{fig:fifa_mad} presents the results where the left panel gives the absolute, and the right panel gives the relative contributions. We see that our expectations concerning the additive variants are confirmed.  
Furthermore, similar to our findings in the simulation study, we see that for the non-additive Euclidean variants, where variable specific effects depend on the number of variables, the scaling of categorical variables as proposed by Hennig-Liao, leads to an overemphasis of distances due to the numerical variables. 

\begin{figure}[!ht]
\center
\includegraphics[scale=.65]{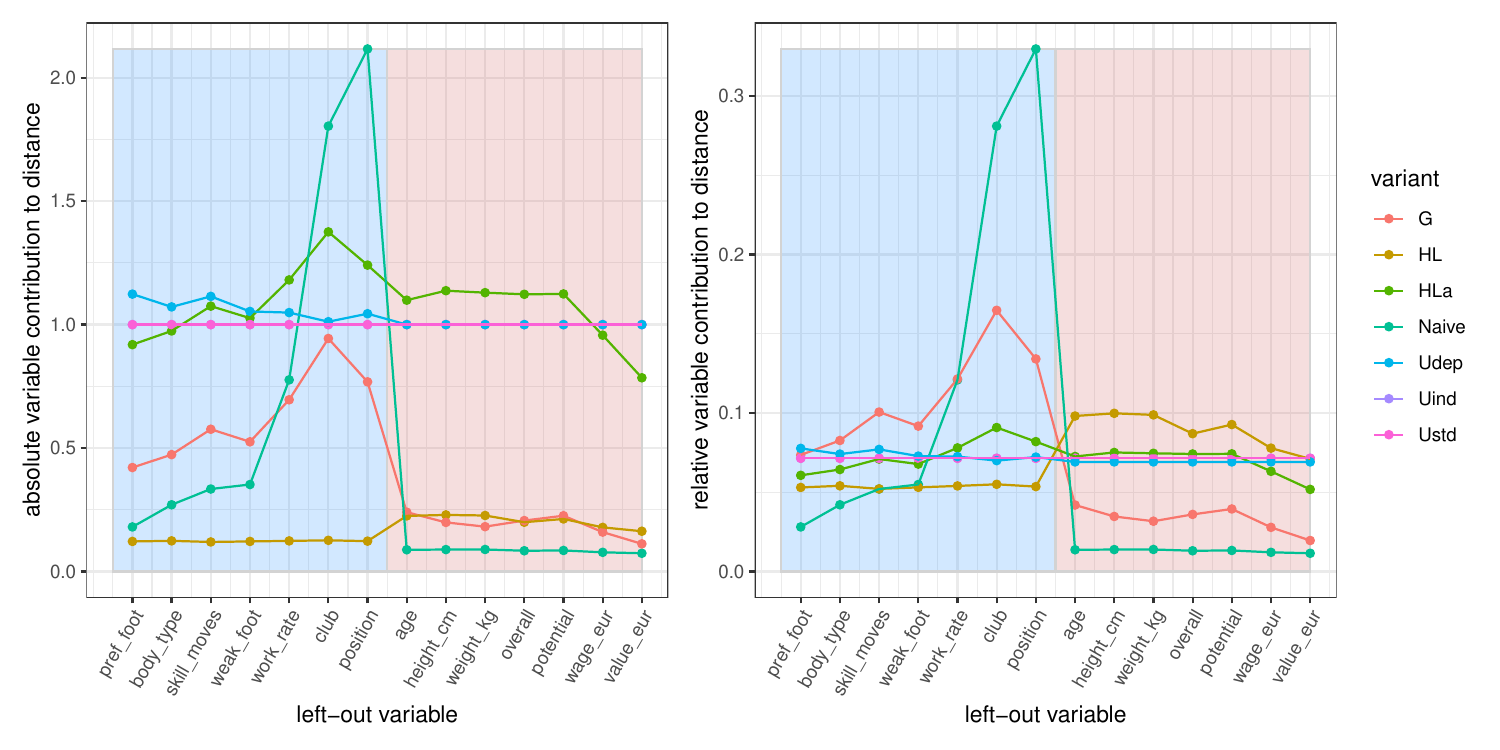}
\caption{Variable importance with respect to full distance matrix. Left panel, absolute values, right panel, relative values}
\label{fig:fifa_mad}
\end{figure}

\begin{figure}[!ht]
\center
\includegraphics[scale=.7]{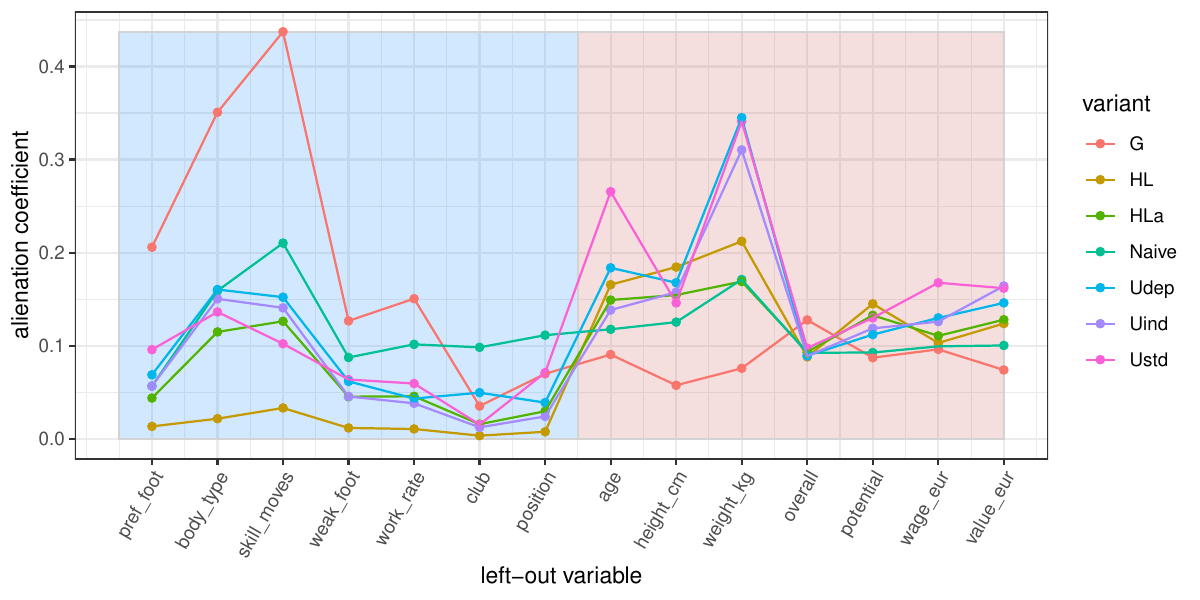}
\caption{Variable importances for alienation coefficient with respect to full distance matrix}
\label{fig:fifa_ac}
\end{figure}

Similar to our analysis of variable specific effects in Section \ref{sect:varspecific}, we consider variable importance with respect to the distances as well as to the $2-$dimensional MDS solutions. Note, however, that there is no true underlying, numerical, configuration. Hence, we can only consider the effect of individual variables with respect to the MDS solution using all variables. The results for variable importance and alienation coefficients are presented in Figures \ref{fig:fifa_mad} and \ref{fig:fifa_ac} where, with the exception of the numerical variant-- which is not possible in this truly mixed setting-- we consider the same variants as in Section \ref{sect:simulation}. 

In Figure \ref{fig:fifa_mad} we see that for the \textit{Naive} variant, distances from categorical variables are much larger than from numerical ones, as expected. Similartly, \textit{Gower} distance,  shows a similar bias. In contrast, \textit{Hennig-Liao} scaling with Euclidean distance over-corrects, making numerical variables overly dominant. We can also see, that a drop in mean values for the last two numerical values. This is due to the skewness of the distributions of these variables. 

Important to note is that although for the unbiased variants the mean distances per variable are equivalent, when it comes to impact on the MDS solutions, there is variability. That is, commensurability does not mean that the variables play the same role in determining an underlying configuration. Or, more generally, that this is the case for any subsequent analysis. Furthermore, the biased variants not only lead to predictable mean variable distances, they also bias the MDS, or any other subsequent, analysis (Figure \ref{fig:fifa_ac}).  

\section{Conclusion} 
In this paper, we considered the problem of biased distance. That is, distances that by construction under- or over-represent contributions of variables based on variable type (e.g., numerical or categorical), distribution (e.g., symmetric or skewed) or scale (e.g., measurement scale or the number of categories). We introduced a general multivariate mixed variable distance that allows for an easy, yet flexible, way to correct for biased distances. Our proposed general distance allows for many different implementations. In particular, concerning distances due to categorical variables, such flexibility is desirable as there is no unique way to define distance for such variables. 

Defining mixed-variable distances has recently received significant attention \parencite{MS:2023,GT:2024, BojGrane:2024}. Most methods propose a variable weighting system \parencite{AD:2007,HennigLiao2013,MS:2023,GT:2024}, or emphasize leveraging both numerical and categorical variable types \parencite{Huang1998,foss2016semiparametric}. Instead of adding another variant, we focus on the theoretical properties of distances for different variable types to construct a mixed distance formulation. This approach aims to address these issues in a general, rather than case- or application-specific, manner. Finally, it is important to emphasize that we do not claim optimality of any kind for our proposed method. In fact, the idea of an unbiased distance is that we treat all variables equivalently, that is, variable types, scales or measurement levels should not trivially impact the distance. If there is reason to believe that for a certain analysis, certain variables are more (less) important than others, this can and should be taken into account. However, in many, especially unsupervised, distance-based methods, there is no objective way of deciding, beforehand, whether this is the case. 

\section*{Source code}
The simulation study and the illustration example presented in this article can be reproduced following the instructions at \url{https://alfonsoiodicede.github.io/blogposts_archive/Unbiased_mixed%20_variables_distance_supplementary_mat.html}. 

\printbibliography

@article{Luo:2024,
author = {Luo, Hengrui and Cho, Younghyun and Demmel, James W. and Li, Xiaoye S.  and Liu, Yang},
title = {Hybrid Parameter Search and Dynamic Model Selection for Mixed-Variable Bayesian Optimization},
journal = {Journal of Computational and Graphical Statistics},
volume = {33},
number = {3},
pages = {855--868},
year = {2024}
}

@article{Amiri:2018,
author = {Amiri, Saeid and Clarke, Bertrand S. and Clarke, Jennifer L.},
title = {Clustering Categorical Data via Ensembling Dissimilarity Matrices},
journal = {Journal of Computational and Graphical Statistics},
volume = {27},
number = {1},
pages = {195--208},
year = {2018}
}

@article{BojGrane:2024,
title = {The robustification of distance-based linear models: Some proposals},
journal = {Socio-Economic Planning Sciences},
volume = {95},
pages = {101992},
year = {2024},
author = {Boj, Eva and Grané, Aurea}
}

@article{Hermes:2024,
author = {Hermes, Sjoerd and van Heerwaarden, Joost and Behrouzi, Pariya},
title = {Copula Graphical Models for Heterogeneous Mixed Data},
journal = {Journal of Computational and Graphical Statistics},
volume = {33},
number = {3},
pages = {991--1005},
year = {2024},
}

@article{van2019distance,
  title={Distance-based clustering of mixed data},
  author={van~de~Velden, Michel and Iodice D'Enza, Alfonso and Markos, Angelos},
  journal={Wiley Interdisciplinary Reviews: Computational Statistics},
  volume={11},
  number={3},
  pages={e1456},
  year={2019}
}

@article{Gower:1966,
 author = {Gower, John C.},
 journal = {Biometrika},
 number = {3/4},
 pages = {325--338},
 title = {Some Distance Properties of Latent Root and Vector Methods Used in Multivariate Analysis},
 volume = {53},
 year = {1966}
}

@article{vandeveldenetal2024,
author = {van~de~Velden, Michel and Iodice D'Enza, Alfonso and Markos, Angelos and Cavicchia, Carlo},
title = {A general framework for implementing distances for categorical variables},
journal = {Pattern Recognition},
volume = {153},
pages = {110547},
year = {2024}
}

@article{de2021towards,
  title={Towards a more balanced combination of multiple traits when computing functional differences between species},
  author={de Bello, Francesco and Botta-Duk{\'a}t, Zolt{\'a}n and Lep{\v{s}}, Jan and Fibich, Pavel},
  journal={Methods in Ecology and Evolution},
  volume={12},
  number={3},
  pages={443--448},
  year={2021},
  publisher={Wiley Online Library}
}

@article{foss2016semiparametric,
  title={A semiparametric method for clustering mixed data},
  author={Foss, Alex and Markatou, Marianthi and Ray, Bonnie and Heching, Aliza},
  journal={Machine Learning},
  volume={105},
  pages={419--458},
  year={2016},
  publisher={Springer}
}

@article{Gower:1971,
 author = {Gower, John C.},
 journal = {Biometrics},
 number = {4},
 pages = {857--871},
 title = {A General Coefficient of Similarity and Some of Its Properties},
 volume = {27},
 year = {1971}
}

@article{HennigLiao2013,
    author = {Hennig, Christian and Liao, Tim F.},
    title = "{How to Find an Appropriate Clustering for Mixed-Type Variables with Application to Socio-Economic Stratification}",
    journal = {Journal of the Royal Statistical Society Series C: Applied Statistics},
    volume = {62},
    number = {3},
    pages = {309-369},
    year = {2013}
}

@article{AD:2007,
title = {A k-mean clustering algorithm for mixed numeric and categorical data},
author = {Ahmad, A. and Dey, L.},
journal = {Data \& Knowledge Engineering},
volume = {63},
number = {2},
pages = {503-527},
year = {2007}
}

@InProceedings{ROBNLH:2015,
author = {Ring, M. and Otto, F. and Becker, M. and Niebler, T. and Landes, D. and Hotho, A.},
editor = {Appice, A. and Rodrigues, P.P. and Santos Costa, V. and Soares, C. and Gama, J. and Jorge, A.},
title = {ConDist: A Context-Driven Categorical Distance Measure},
booktitle = {Machine Learning and Knowledge Discovery in Databases},
year = {2015},
publisher = {Springer International Publishing},
address = {Cham},
pages = {251--266}
}

@article{JCL:2014,
title = {A New Distance Metric for Unsupervised Learning of Categorical Data},
author = {Jia, H. and Cheung, Y. and Liu, J.},
journal = {IEEE Transactions on Neural Networks and Learning Systems},
year = {2014},
volume = {27},
pages = {1065-1079}
}

@article{LH:2005,
title = {An association-based dissimilarity measure for categorical data},
journal = {Pattern Recognition Letters},
volume = {26},
number = {16},
pages = {2549-2557},
year = {2005},
issn = {0167-8655},
author = {Le, S.Q. and Ho, T.B.},
}

@inproceedings{BCK:2008,
title = {Similarity measures for categorical data: A comparative evaluation},
author = {Boriah, S. and Chandola, V. and Kumar, V.},
booktitle = {Proceedings of the 2008 SIAM international conference on data mining},
pages = {243--254},
year = {2008},
organization = {SIAM}
}

@Manual{fpc_pkg,
    title = {fpc: Flexible procedures for clustering},
    author = {Christian Hennig},
    year = {2024},
    url = {https://CRAN.R-project.org/package=fpc},
    note = {R package version 2.2-12}
  }

@article{SR:2019,
title={Comparison of similarity measures for categorical data in hierarchical clustering},
author={{\v{S}}ulc, Z. and {\v{R}}ezankov{\'a}, H.},
journal={Journal of Classification},
volume={36},
number={1},
pages={58--72},
year={2019}
}

@article{BMMFRBC:2023,
title = {Do all roads lead to {R}ome? {S}tudying distance measures in the context of machine learning},
journal = {Pattern Recognition},
volume = {141},
pages = {109646},
year = {2023},
author = {Blanco-Mallo, E. and Morán-Fernández, L. and Remeseiro, B. and Bolón-Canedo, V.}
}

@Manual{cluster_pkg,
    title = {cluster: Cluster Analysis Basics and Extensions},
    author = {Martin Maechler and Peter Rousseeuw and Anja Struyf and
      Mia Hubert and Kurt Hornik},
    year = {2023},
    url = {https://CRAN.R-project.org/package=cluster},
    note = {R package version 2.1.6 --- For new features, see the
      'NEWS' and the 'Changelog' file in the package source)}
  }

@Manual{gower_pkg,
    title = {gower: Gower’s Distance},
    author = {Mark van der Loo},
    year = {2017},
    url = {https://CRAN.R-project.org/package=gower},
    note = {R package version 0.1.2.}
  }

@article{Huang1998,
  title={Extensions to the k-Means Algorithm for Clustering Large Data Sets with Categorical Values},
  author={Joshua Zhexue Huang},
  journal={Data Mining and Knowledge Discovery},
  year={1998},
  volume={2},
  pages={283-304}
}

@article{MS:2023,
title = {A generalized multi-aspect distance metric for mixed-type data clustering},
author = {Mousavi, E and Sehhati, M},
journal = {Pattern Recognition},
volume = {138},
pages = {109353},
year = {2023}
}

@article{BL:1985,
    author = {Borg, I and Leutner, D},
    title = {Measuring the Similarity of MDS Configurations},
    journal = {Multivariate Behavioral Research} ,
    year = {1985},
volume ={20},
pages={325–334},
}

@book{BG:2020,
  title={Hands-on machine learning with R},
  author={Boehmke, Brad and Greenwell, Brandon M},
  year={2020},
  publisher={Chapman and Hall/CRC}
}

@article{GT:2024,
 author = {Ghashti, J.S. and Thompson, J.R.J.},
 title= {Mixed-Type Distance Shrinkage and Selection for Clustering via Kernel Metric Learning},
 journal= {Journal of Classification},
 year = {2024},
 doi= {https://doi.org/10.1007/s00357-024-09493-z}
}

\appendix
\section{Appendix}
\subsection{Mean distance results: Uniform distribution}\label{appendix:exactresultsU}
Let $X \sim U\left(a,b\right)$ denote a random variable from a uniform distribution on the interval $[a,b]$, where $a$ and $b$ are known. Furthermore, let $X^{s}, X^{rn}$, and $X^{rr}$ denote the standardized, range normalized and robust range normalized versions of $X$. 
Then:
\begin{theorem}
  $E[\lvert X^{s}_{i} -X^{s}_{l}\rvert ]=c=\frac{1}{3}\sqrt{12} (\approx 1.15)$
  \end{theorem}
\begin{theorem}
       $E[\lvert X^{rn}_{i} -X^{rn}_{l}\rvert ]=c=\frac{1}{3}$
  \end{theorem}
\begin{theorem}
        $E[\lvert X^{rr}_{i} -X^{rr}_{l}\rvert ]=c=\frac{2}{3}$
  \end{theorem} 
\begin{proof}
  Standardization gives
\begin{align*}
    X^{s} = \frac{X^{s}-\mu}{\sigma_{x}}= \frac{X^{s}-(a+(b-a)/2)}{(b-a)/\sqrt(12)} \sim U \left(-\frac{\sqrt{12}}{2},\frac{\sqrt{12}}{2}\right)
\end{align*}
The expected value of the absolute difference between two independent uniform random variables on $[a,b]$ is $(b-a)/3$. Hence 
\begin{equation*}
    E[\lvert X^{s}_{i} -X^{s}_{l}\rvert ]=c=\frac{1}{3}\sqrt{12}.
\end{equation*} 
\end{proof} {\hfill $\blacksquare$}

\begin{proof}
Range normalization gives \begin{align*}
    X^{rn}= \frac{X^{rn}-a}{\left( b-a \right)} \sim U \left(0,1 \right),
\end{align*}
so that
\begin{equation*}
    E[\lvert X^{rn}_{i} -X^{rn}_{l}\rvert ]=c=\frac{1}{3}.
\end{equation*}
\end{proof} {\hfill $\blacksquare$}
\begin{proof}
Subtracting the median and dividing by the interquartile range for a $U\left(a,b\right)$ random variable gives
\begin{align*}
    X^{rr}= \frac{X^{rr}-(a+(b-a)/2)}{\left(b-a \right)/2} \sim U \left(\frac{2a}{b-a}, \frac{2b}{b-a}\right).
\end{align*}
Resulting in
\begin{equation*}
    E[\lvert X^{rr}_{i} -X^{rr}_{l}\rvert ]=c=\frac{2}{3}.
\end{equation*}
 \end{proof} {\hfill $\blacksquare$}
 
\section{Mean distance results: Normal distribution}\label{appendix:exactresultsN}
Let  $X \sim N\left(\mu,\sigma\right)$ denote a random variable from a normal distribution with known mean $\mu$ and standard deviation $\\sigma$. Furthermore, let $X^{s}, X^{rn}$, and $X^{rr}$ denote the standardized, range normalized and robust range normalized versions of $X$. We have
\begin{theorem}
    $E[\lvert X^{s}_{i} -X^{s}_{l} \rvert ]=2\sqrt{\left(1/\pi\right)} (\approx 1.13)$
\end{theorem}
\begin{theorem}
   $ E[\lvert X^{rn}_{i} -X^{rn}_{l} \rvert ]=\frac{1}{x^{0.75}} \sqrt{\left(1/\pi\right)} (\approx 0.84)$
\end{theorem}
\begin{proof}
The distribution of the difference between two identically distributed normal random variables is also normal with mean zero and standard deviation $\sqrt{2}\sigma$. Furthermore, the absolute value of a centered normal random variable, say $Y$ is \say{ folded normal} with as its expected value
\begin{equation}\label{EV_fNormal}
    E[\lvert Y\rvert ]=\sigma\sqrt{\left(2/\pi\right)}.
\end{equation}
Using standard deviation standardization (assuming known $\sigma$) we have
\begin{align*}
    X^{s}= \frac{X-\mu}{\sigma} \sim N \left(0,1\right).
\end{align*}
Hence, 
\begin{equation*}
    E[\lvert X^{s}_{i} -X^{s}_{l} \rvert ]=2\sqrt{\left(1/\pi\right)}.
\end{equation*}
\end{proof} {\hfill $\blacksquare$}
\begin{proof}
For the difference between normal random variables using robust range scaling, we have
    \begin{align*}
    X^{rn}_i - X^{rn}_l \sim N\left( 0,\sqrt{2}\frac{\sigma}{\operatorname {IQR} }\right).
\end{align*}
As the interquartile range $\operatorname{IQR}$ is linearly related to $\sigma$, we can set, without loss of generality, $\sigma=1$ and use the interquartile range of a standard normal distribution. That is, $\operatorname{IQR}= 2x^{0.75}$, where $P(X<x^{0.75})=0.75$. 

As before, taking absolute values leads to a folded normal distribution. Inserting the appropriate expressions into \ref{EV_fNormal} yields 
\begin{equation*}
    E[\lvert X^{rn}_{i} -X^{rn}_{l} \rvert ]=\frac{1}{x^{0.75}} \sqrt{\left(1/\pi\right)}.
\end{equation*}

\end{proof} {\hfill $\blacksquare$}

\section{Hennig-Liao limiting results}
Let $T$ denote the total potential pairwise differences among all observations. In the function \code{distancefactor} from the \proglang{R} package \pkg{fpc}, this is calculated as $T = n(n + 1)/2$. Let $W$ denote the cumulative \say{within-category} difference. This accounts for the number of ways two observations from the same category can be chosen, and is defined as
        \[
        W = \sum_{a \in \{1 \dots q\}} { \frac{n_a\left(n_a+1\right)}{2} },
        \]
        where $n_a$ denotes the number of observations with category $a$.
Define $B$ as the difference between the total and withing category differences. That is,
\[
B=T-W.
\]
The normalization or scaling factor $\eta$ is defined as
\begin{equation}\label{Eta}      \eta = \sqrt{\phi    \frac{T}{B} },
        \end{equation}
where $\phi$ is user-supplied with as suggested values $1/2$ for categorical and $(q-1)/q$ for ordinal variables. 
The normalization factor $\eta$ depends on the sample distribution, the user-supplied $\phi$, the number of categories $q$ and the number of observations $n$. If all categories are equally likely, we can derive limiting results with respect to the sample size. In particular:
\begin{theorem}\label{EtaLimits}
\[ \lim_{n\to\infty} \eta(n,q,\phi) = \sqrt{\phi\frac{q}{(q-1)}}, \]
\[ \lim_{n\to q} \eta(n,q,\phi) = \sqrt{ \phi\frac{(q+1)}{(q-1)}}. \]

\end{theorem}

\begin{proof}
As each category is equally likely, the number of observations per category converges to  $n_a = n/q$, for all $a \in \{1 \dots q\}$. Furthermore, for $n$ large, $n_a +1 \to n_a$ and $n+1 \to n$. Consequently, $T=n(n+1)/2 \to n^2/2$ and $W =  qn_a\left(n_a+1\right)/2 \to qn_a^2/2$. Hence, $B \to n_a^2q(q-1)/2$. Inserting the limiting results for $T$ and $B$ into \ref{Eta} gives the first limiting result of Theorem \ref{EtaLimits}. The second limit result can be verified directly by inserting $n=q$ and $n_a=1$ into \ref{Eta}.
\end{proof} {\hfill $\blacksquare$}

\noindent When we have equal probabilities and $n$ goes to infinity, the variance per variable equals 
\begin{equation*}
    V(Z^{*}) = qV(Z^{*}_a)=q \eta^2 V(Z_a) = q\eta^2 \frac{1}{q}\left(1-\frac{1}{q}\right) = \eta^2 \left(\frac{q-1}{q}\right)= \phi. 
\end{equation*}
\noindent Hence, equation (\ref{HLscaling}) is satisfied.

\section{Categorical mean distances: skewed distributions}\label{appendix:skewedistributions}
The expected mean distances for the categorical variables distances depend on the chosen distance as well as on the distributions $\bm{p}$. Table \ref{table:expectedvalues} gives results when all categories are equally likely. Here we consider distributions where $p_1 \in \{0.05,0.1,0.2,0.33, 0.5,0.66, 0.8,0.9,0.95\}$ and $p_j = (1-p_1)/(q-1)$, with $j=2,\dots,q$, for $q \in \{2,3,5,10\}$. For the association-based variants we again consider the perfect bivariate dependence case where the marginal probabilities coincide and the number of categories for the two variables is equal.

Figures \ref{fig:expvals1} through \ref{fig:exvals_assbased} summarize the results for different numbers of categories and different levels of skewness for the distances from Table \ref{table:expectedvalues}. Note that the scales for the different distances vary significantly. For example, the values for the OF and IOF variants are much larger than other values. Moreover, for the IOF distance, the number of observations is also of importance. Here we selected $n=160$ for all cases. Increasing this would lead to larger values. Decreasing the value can lead to negative results for some scenarios as $np$ becomes smaller than $1$. 

\begin{figure}[!ht]
\center
\includegraphics[scale=.6]{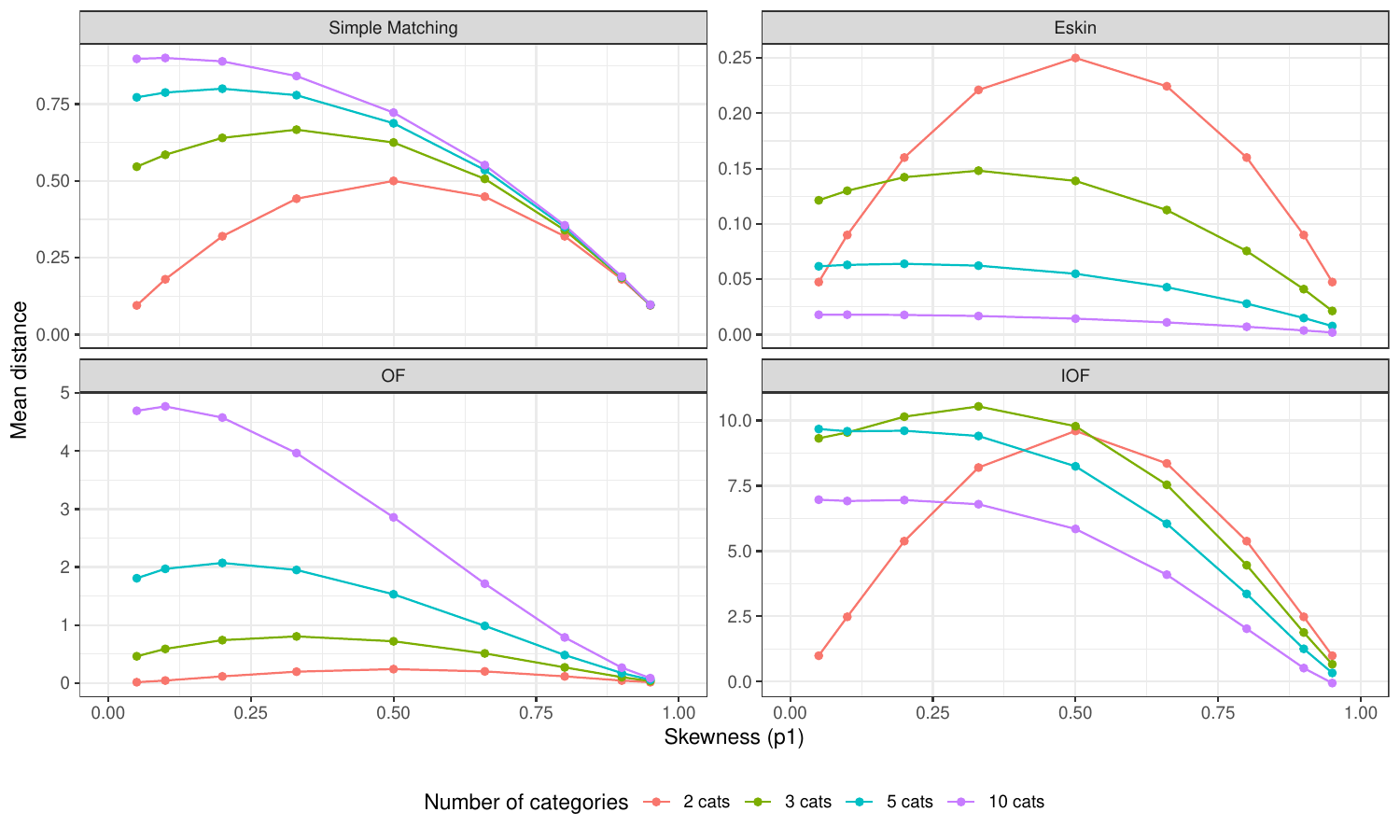}
\caption{Mean values for independent categorical variable distances}
\label{fig:expvals1}
\end{figure}

\begin{figure}[!ht]
\center
\includegraphics[scale=.6]{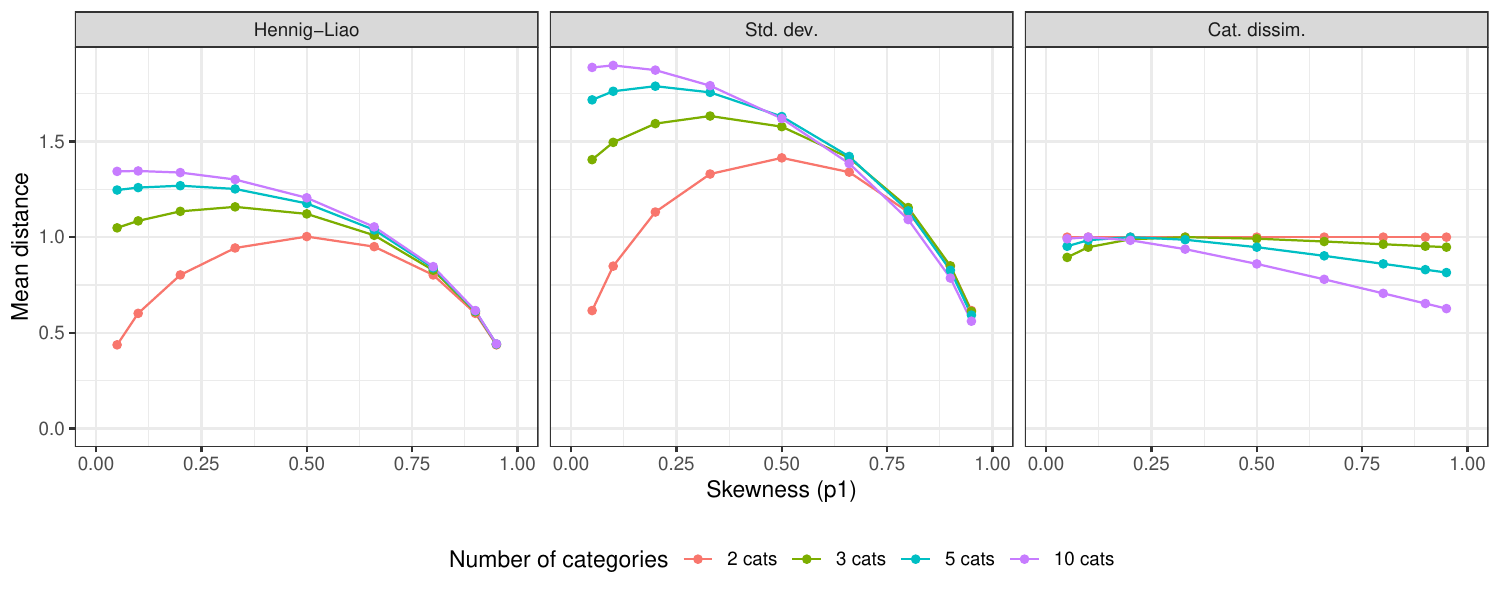}
\caption{Mean values for independent categorical variable distances}
\label{fig:expvals2}
\end{figure}

\begin{figure}[!ht]
\center
\includegraphics[scale=.65]{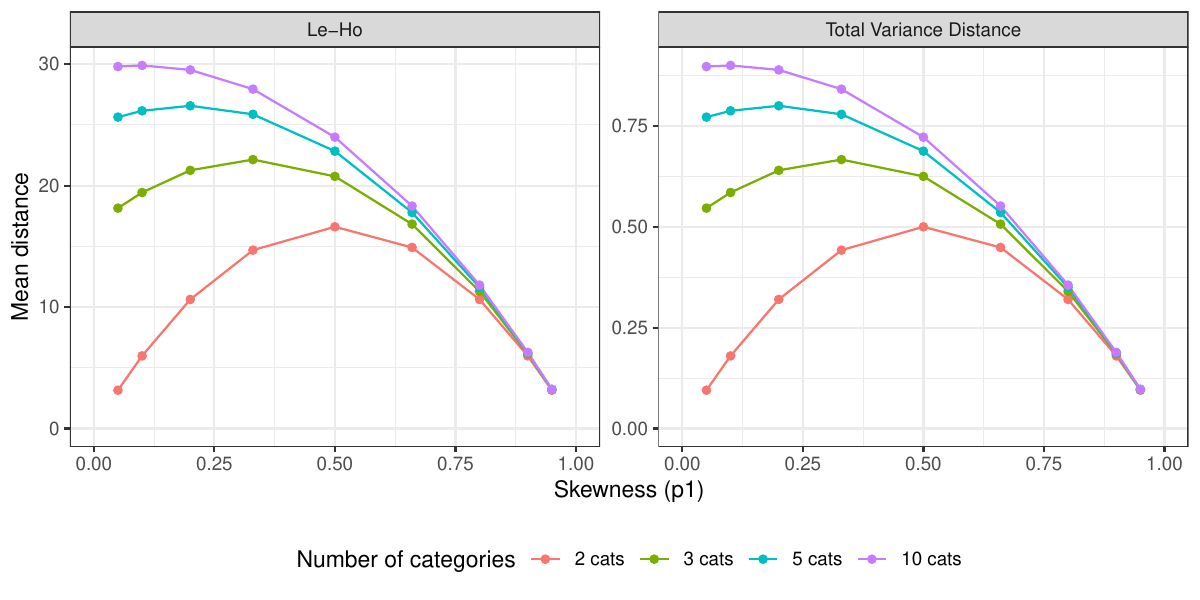}
\caption{Mean values for association-based distances based on dependent bi-variate data}\label{fig:exvals_assbased}
\end{figure} 

We see that for all variants the distribution as well as the number of categories influences the mean distances. 

Figure \ref{fig:exvals_assbased} illustrates the effect of the underlying marginal distributions for the dependent scenario. We see that the absolute size of the mean values differs considerably between methods. In particular, the values for the Kullback-Leibler variant are much larger. The pattern, however, is quite similar and for both variants the maximum mean distance occurs when all categories are equally likely, indicating maximum variance.  

\section{FIFA variable distributions }\label{appendix:FIFA}
\begin{figure}[!ht]
\center
\includegraphics[scale=.55]{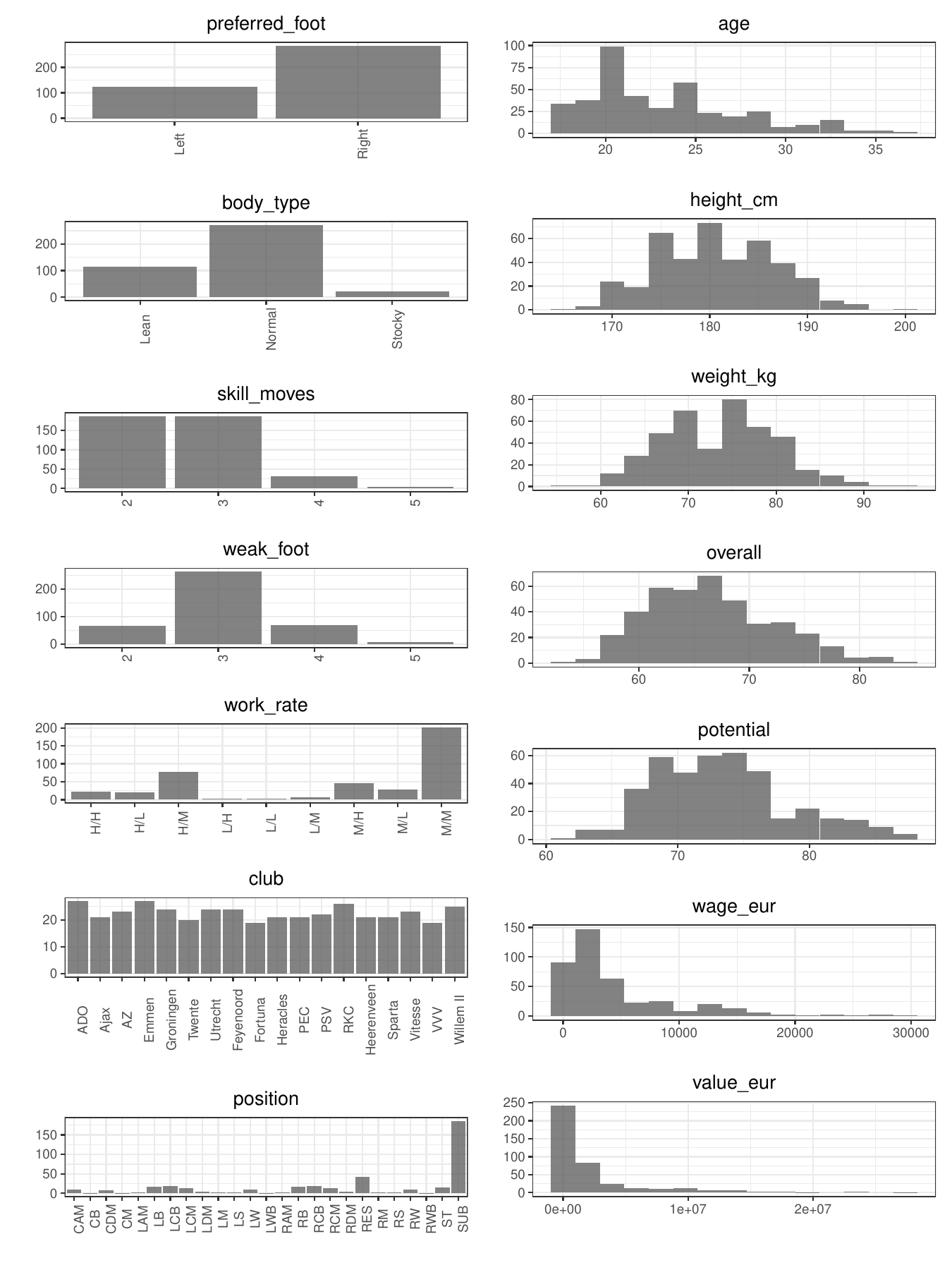}
\caption{Bar plots for the FIFA categorical variables (left), histograms for numerical variables (right)}\label{fig:fifa_desc}
\end{figure} 
\end{document}